\documentclass[12pt]{article}     

\usepackage{a4wide}
\usepackage{amsmath}
\usepackage{amssymb}

\usepackage{enumerate}

\usepackage{amsfonts}

\usepackage[all]{xy}

\usepackage{wasysym}
\usepackage{stackengine} 

\newcommand{\newcommandm}[2]{\newcommand{#1}{{\ensuremath{#2}}}}

\DeclareFontFamily{OT1}{pzc}{}
\DeclareFontShape{OT1}{pzc}{m}{it}{<-> s * [1.10] pzcmi7t}{}
\DeclareMathAlphabet{\mathpzc}{OT1}{pzc}{m}{it}

\newcommand{\powerset}[1]{\mathcal{P}({#1})}

\newcommand{\integers}{\mathbb{Z}}
\renewcommand{\phi}{\varphi}

\newcommandm{\agents}{A}
\newcommand{\agentA}{a}
\newcommand{\agentB}{b}

\newcommandm{\atoms}{P}
\newcommand{\atomP}{p}
\newcommand{\atomQ}{q}

\newcommand{\atomsQ}{Q}

\newcommand{\state}[2][\empty]{{{#2}^{#1}}}
\newcommand{\stateS}[1][\empty]{\state[#1]{s}}
\newcommand{\stateT}[1][\empty]{\state[#1]{t}}
\newcommand{\stateU}[1][\empty]{\state[#1]{u}}
\newcommand{\stateV}[1][\empty]{\state[#1]{v}}
\newcommand{\stateW}[1][\empty]{\state[#1]{w}}
\newcommand{\stateX}[1][\empty]{\state[#1]{x}}

\newcommand{\statesT}[1][\empty]{\state[#1]{T}}

\newcommand{\states}[1][\empty]{\ensuremath{{S^{#1}}}}
\newcommand{\accessibility}[2][\empty]{\ensuremath{{\sim^{#1}_{#2}}}}
\newcommand{\successors}[3][\empty]{[{#3}]^{#1}_{#2}}

\newcommand{\valuation}[1][\empty]{\ensuremath{{V^{#1}}}}
\newcommand{\interpretation}[2][\empty]{[\![{#2}]\!]_{#1}}

\newcommand{\modelTuple}[1][\empty]{\ensuremath{{(\states[#1], \accessibility[#1]{}, \valuation[#1])}}}
\newcommand{\model}[1][\empty]{\ensuremath{{M^{#1}}}}
\newcommand{\modelAndTuple}[1][\empty]{\model[#1] = \modelTuple[#1]}

\newcommand{\pointedModel}[2][\empty]{\ensuremath{{M^{#1}_{#2}}}}
\newcommand{\pointedModelTuple}[2][\empty]{\ensuremath{{(\modelTuple[#1], {#2})}}}
\newcommand{\pointedModelAndTuple}[2][\empty]{\pointedModel[#1]{#2} = \pointedModelTuple[#1]{#2}}

\newcommand{\bisimilar}[1][\empty]{\simeq^{#1}}
\newcommand{\refines}[1][\empty]{\preceq_{#1}}
\newcommand{\simulates}[1][\empty]{\succeq_{#1}}

\newcommand{\bisimulation}{\mathfrak{R}}
\newcommand{\refinement}{\bisimulation}

\newcommand{\entails}{\models}  

\newcommand{\nentails}{\not\models} 

\renewcommand{\implies}{\rightarrow}
\renewcommand{\iff}{\leftrightarrow}

\newcommand{\necessary}[1][\empty]{{K_{#1}}}
\newcommand{\possible}[1][\empty]{{L_{#1}}}

\newcommand{\knows}{\necessary}
\newcommand{\suspects}{\possible}

\newcommand{\restrict}[2]{{{#1} | {#2}}}
\newcommand{\announceE}[1]{\langle #1 \rangle}
\newcommand{\announceA}[1]{[ #1 ]}

\newcommand{\somepas}{\Diamond}
\newcommand{\allpas}{\Box}
\newcommand{\someppas}{{\ensurestackMath{\stackinset{c}{-0.01ex}{c}{0.225ex}{\scriptscriptstyle+}{\raisebox{-0.225ex}{$\Diamond$}}}}}
\newcommand{\allppas}{{\ensurestackMath{\stackinset{c}{}{c}{}{\scriptscriptstyle+}{\raisebox{-0.225ex}{$\square$}}}}}

\newcommandm{\lang}{\mathpzc{L}} 
\newcommandm{\langPl}{\lang_\mathit{pl}} 
\newcommandm{\langMl}{\lang_\mathit{el}} 
\newcommandm{\langMlPlus}{\lang_\mathit{el}^{+}} 

\newcommandm{\langPal}{\lang_\mathit{pal}} 
\newcommandm{\langApal}{\lang_\mathit{apal}} 
\newcommandm{\langPapal}{\lang_\mathit{apal}^{+}} 
\newcommandm{\langPapalt}{\lang_\mathit{apal^{+*}}} 

\newcommandm{\logic}{\mathit{L}}
\newcommandm{\classS}{\mathpzc{S5}}
\newcommandm{\logicS}{\mathit{S5}}
\newcommandm{\logicPal}{\mathit{PAL}}
\newcommandm{\logicApal}{\mathit{APAL}}
\newcommandm{\logicBapal}{\mathit{BAPAL}}
\newcommandm{\logicPapal}{\mathit{APAL}^+}
\newcommandm{\axiomS}{\mathbf{S5}}
\newcommandm{\axiomPal}{\mathbf{PAL}}
\newcommandm{\axiomApal}{\mathbf{APAL}}
\newcommandm{\axiomPapal}{\mathbf{APAL}^+_\omega}
\newcommandm{\axiomPapalF}{\mathbf{APAL}^+_1}

\usepackage{tikz}
\usetikzlibrary{arrows}

\usepackage[thmmarks]{ntheorem}
\usepackage{newproof}

\newtheorem{theorem}{Theorem}[section]  
\newtheorem{corollary}[theorem]{Corollary}

\newtheorem{definition}[theorem]{Definition}
\newtheorem{lemma}[theorem]{Lemma}
\newtheorem{proposition}[theorem]{Proposition}

\usepackage{algorithm,algorithmic}

\newcommand{\eq}{\leftrightarrow}
\newcommand{\Eq}{\Leftrightarrow}
\newcommand{\imp}{\rightarrow}

\newcommand{\vel}{\vee}

\newcommand{\Vel}{\bigvee}

\newcommand{\Dia}{\Diamond}
\newcommand{\dia}[1]{\langle #1 \rangle}
\renewcommand{\phi}{\varphi}
\newcommand{\union}{\cup}

\newcommand{\inter}{\cap}

\newcommand{\et}{\ensuremath{\land}}

\newcommand{\weg}[1]{}

\usepackage{url}
\newcommand{\Naturals}{\mathbb N}
\newcommand{\domain}{\mathcal{D}}

\begin{document}

\author{Hans van Ditmarsch\thanks{CNRS, LORIA, Univ.\ of Lorraine, France \& IMSc, Chennai, India, {\tt hans.van-ditmarsch@loria.fr}. Corresponding author. We kindly acknowledge support from ERC project EPS 313360. A preliminary version of this work was presented at the 2014 Software Engineering Research Conference organized by the UWA doctoral school, see {\tt http://jameshales.org/manuscripts/csse2014.pdf}.}, Tim French\thanks{University of Western Australia, Perth, Australia, {\tt tim.french@uwa.edu.au}}, James Hales\thanks{University of Western Australia, Perth, Australia, {\tt james@jameshales.org}}}
\title{Positive Announcements}
\date{}

\maketitle

\begin{abstract}
    {\em Arbitrary public announcement logic} (\logicApal) reasons about
    how the knowledge of a set of agents changes after true public announcements and after arbitrary announcements of true epistemic formulas. We consider a
    variant of arbitrary public announcement logic called {\em positive arbitrary
    public announcement logic} (\logicPapal), which restricts arbitrary public
    announcements to announcement of {\em positive formulas}. Positive formulas prohibit statements
    about the ignorance of agents. The positive formulas correspond to the universal fragment in first-order logic. As two successive announcements of positive formulas need not correspond to the announcement of a positive formula, \logicPapal{} is rather different from \logicApal.  We show that \logicPapal{} is more
    expressive than public announcement logic \logicPal, and that \logicPapal\ is
    incomparable with \logicApal. We also provide a sound and complete infinitary axiomatisation. \\ 


\end{abstract}

\section{Introduction and overview} 

Public announcement logic (\logicPal{})~\cite{gerbrandyetal:1997,plaza:1989} extends
epistemic logic with operators for reasoning about the effects of specific
public announcements. The formula $\announceA{\psi} \phi$ means that ``$\phi$
is true after the truthful announcement of $\psi$''. This means that, when interpreted in an epistemic model with designated state, after submodel restriction to the states where $\psi$ is true (this includes the designated state, and `truthful' here means true), $\phi$ is true in that restriction. Arbitrary public
announcement logic (\logicApal{})~\cite{balbianietal:2008} augments this with
operators for quantifying over public announcements. The formula
$\allpas \phi$ means that ``$\phi$ is true after the truthful announcement of
any formula that does not contain $\Box$''. 

Quantifying over the communication of information as in
\logicApal{} has applications to epistemic protocol synthesis, where
we wish to achieve epistemic goals by communicating information to agents, but
where we do not know of a specific protocol that will achieve the goal, and where we may not
even know if such a protocol exists. In principle, synthesis problems can be solved
by specifying them as formulas in the logic, and applying model-checking or
satisfiability procedures. However in the case of \logicApal{}, while there is
a PSPACE-complete model-checking procedure~\cite{agotnesetal.jal:2010}, the 
satisfiability problem is undecidable in the presence of multiple agents~\cite{frenchetal:2008}. 

We consider a variant of \logicApal{} called {\em positive arbitrary public
announcement logic} (\logicPapal{}), we obtain various semantic results relating refinements to positive formulas, we give various rather surprising expressivity results, and we give a non-surprising axiomatization. In \logicApal{} the arbitrary public announcements quantify over quantifier-free formulas, that are equivalent to epistemic formulas (basic modal logic). Whereas in \logicPapal{} the arbitrary public announcements quantify over quantifier-free {\em positive} formulas: formula $\allppas \phi$ means that
``$\phi$ is true after the truthful public announcement of any {\em positive}
formula''. A formula is {\em positive} if, roughly, the knowledge modalities are never bound by negations. Positive formulas consist only of positive knowledge statements,
such as ``it is known that'', and prohibit negative knowledge statements such as
``it is not known that'' and ``it is uncertain that''. In the standard translation, such formulas correspond to the universal fragment \cite{andrekaetal:1998}. 

The restriction to positive formulas is natural in view of possible applications. There are many protocols wherein the messages convey that an 
agent {\em knows} an atomic proposition and wherein only the invariants or postconditions require that an agent {\em does not know} an atomic proposition. Knowledge of atomic propositions is stable and easy to verify whereas absence of knowledge is fragile and, typically, hard to verify. For example, verifying knowledge is done by direct observation such as witnessing a communication, or by message passing between principals in a security protocol (where messages are considered atomic components), or by reading a time-stamped blockchain ledger \cite{blockchain}. However, verifying that an agent does not know a proposition requires an assumption that there are no private communication channels or clandestine messages, and thus negative knowledge cannot be verified in the same way as positive knowledge. Consequently, quantifying over positive announcements can often be viewed as quantifying over protocols consisting of straightforwardly verifiable information. The decidability of positive arbitrary public announcement logic therefore means that we can answer the question whether it is possible to achieve a particular knowledge state by means of such protocols. 

Let us give some other concrete examples. In the 
alternating bit protocol \cite{halpernzuck:1992} the communicating agents achieve partial correctness of message transfer by stacking acknowledgements (where `acknowledge' means `know'). The internet protocol TCP/IP manages package transfer, taking into account of missing packages and time-outs, again by means of stacked knowledge \cite{stulpetal:2002}. In those case there are no concerns involving ignorance, it is a matter of guaranteeing (partial) knowledge. In various security protocols the worst-case scenario is that all messages between principals are intercepted, in other words, that they become public announcements (all aspects of the protocol except private keys may be assumed public). For example, in cards cryptography two communicating agents attempt to learn the card deal without other players (eavesdroppers) learning the card deal (or even any single card other than their own) \cite{fischeretal:1996,hvd.studlog:2003,cordonetal.tcs:2013}. The dining cryptographers protocol \cite{Chaum:1988,meyden:2004} has semi-public (coin tossing, observed by an agent and its neighbour) and public aspects (announcing bits, depending on the outcome of the coin toss and whether the agent paid for the meal), in order to guarantee an ignorance epistemic goal (who paid for the meal?). The public part consists of positive announcements (namely of known values of bits).

The logic \logicPapal{} is decidable. The proof of this result is substantial and of a fairly technical nature and it is therefore reported in a companion paper \cite{papaldec:2018}. As this result puts \logicPapal{} in perspective to similar logics, let us summarily sketch the picture. For an in-depth discussion we refer to \cite{papaldec:2018}. With respect to other logics with quantification over announcements,
\logicApal, the related {\em group announcement logic}, and {\em coalition announcement logic} are all undecidable \cite{agotnesetal:2016} (and all three are only known to have infinitary axiomatisations), whereas the {\em `mental model'} arbitrary public announcement logic of \cite{charrieretal:2015} and {\em Boolean arbitrary public announcement logic} (\logicBapal) \cite{hvdetal.bapal:2017} are decidable. 

As the name suggests, \logicBapal{} has quantification over Boolean announcements \cite{hvdetal.bapal:2017}. This form of quantification is therefore even more restricted than in \logicPapal. Its axiomatisation is finitary, unlike \logicPapal{}, for which we only report an infinitary axiomatisation.

From the dynamic epistemic logics that are quantifying over non-public information change, {\em arbitrary arrow update logic} \cite{hvdetal.undecidable:2017} is undecidable, whereas the already mentioned {\em refinement modal logic} \cite{bozzellietal.inf:2014} and {\em arbitrary action model logic} \cite{hales2013arbitrary} are decidable. For the last two logics this is an elementary consequence of the fact that they are as expressive as the base modal logic. This is shown with respect to ${\mathcal K}$ models (models for arbitrary accessibility relations). In \cite{hales.aiml:2012} it is also shown that refinement modal logic interpreted on models of the class \classS{} (where all accessibility relations are equivalence relations; the logic is then called {\em refinement epistemic logic}) is as expressive as the modal logic $S5$.

We hope that the logic \logicPapal{} offers a valuable contribution to this already diverse landscape of logics with quantification over information change.

In Section \ref{technical-preliminaries} we give an overview of structures and structural notions, such as epistemic model, bisimulation, and refinement, and we present public announcement logic and arbitrary public announcement logic. In Section \ref{syntax-semantics} we give the syntax and semantics of positive arbitrary public announcement logic  \logicPapal. In Section \ref{sec.modelchecking} we show that \logicPapal{} model checking is PSPACE-complete. In Section \ref{sec.expressivity} we demonstrate that \logicApal{} and \logicPapal{} are incomparable. In Section \ref{axiomatisation} we give the complete infinitary axiomatisation of \logicPapal.

\section{Public announcement logics}\label{technical-preliminaries}

We recall definitions and technical results from epistemic logic,
public announcement logic~\cite{gerbrandyetal:1997,plaza:1989} and
arbitrary public announcement logic~\cite{balbianietal:2008}. Throughout this contribution, let $\agents$ be a countable set of \emph{agents} and
let $\atoms$ be a countable set of \emph{propositional atoms} (or \emph{atoms}, or \emph{propositional variables}).

\subsection{Structural notions}

In this subsection we define {\em epistemic models}, {\em model restrictions}, and various types of {\em bisimulation}.
\begin{definition}\label{epistemic-model}
    An {\em epistemic model} $\modelAndTuple$ consists of 
    a {\em domain}  \states{}, 
    which is a non-empty set of states, 
    a set of {\em accessibility relations} \accessibility{}, indexed by agents
    $\agentA \in \agents$, where $\accessibility{\agentA} \subseteq \states
    \times \states$ is an equivalence relation on states  (a relation that is
    reflexive, transitive and symmetric), and a {\em valuation} $\valuation :
    \states \to \powerset \atoms$, which is a function from states to subsets of propositional atoms (namely those true in that state).

    The {\em class of all epistemic models} is called \classS{}.
    A {\em pointed epistemic model} $\pointedModelAndTuple{\stateS}$ 
    consists of an epistemic model \model{} along with a designated state $\stateS \in \states$. A pointed epistemic model will often also be called an epistemic model.
\end{definition}

Given two states $\stateS, \stateT \in \states$, 
we write $\stateS \accessibility{\agentA} \stateT$ to denote that 
$(\stateS, \stateT) \in \accessibility{\agentA}$. 
We write $\successors{\agentA}{\stateS}$ to denote the $\agentA$-equivalence
class of $\stateS$, which is the set of states $\successors{\agentA}{\stateS} =
\{\stateT \in \states \mid \stateS \accessibility{\agentA} \stateT \}$.
As we will often be required to discuss several models at once, we will use
the convention that 
$\pointedModel{\stateS} = \pointedModelTuple{\stateS}$,
$\pointedModel[\prime]{\stateS[\prime]} = \pointedModelTuple[\prime]{\stateS[\prime]}$,
$\pointedModel[\gamma]{\stateS[\gamma]} = \pointedModelTuple[\gamma]{\stateS[\gamma]}$,
etc. If $s \sim_a t$, we say that there is an $a$-\emph{link} ($a$-\emph{step}) between $s$ and $t$. An epistemic model is {\em connected} if between any two states in its domain there is a path consisting of such links, i.e., if for any states $s,t$ there are states $s = s_1, s_2, \dots, s_n = t$ and agents $a_1, \dots, a_{n-1}$ such that for all $1 \leq i \leq n-1$, $s_i \sim_{a_i} s_{i+1}$.

\begin{definition}\label{model-restriction}
    Let $\modelAndTuple \in \classS$ be an epistemic model and $\statesT \subseteq \states$ where $\emptyset\neq T$.
    We define the {\em restriction of $\model$ to $\statesT$} 
    as $\restrict{\model}{\statesT} = (\restrict{\states}{\statesT}, \restrict{\accessibility{}}{\statesT}, \restrict{\valuation}{\statesT})$ where:
    \begin{eqnarray*}
        \restrict{\states}{\statesT} &=& \statesT\\
        \restrict{\accessibility{\agentA}}{\statesT} &=& \accessibility{\agentA} \cap (\statesT \times \statesT)\\
        \restrict{\valuation}{\statesT}(\atomP) &=& \valuation(\atomP) \cap \statesT
    \end{eqnarray*}
\end{definition}
If $N$ is a restriction of $M$ we write $N \subseteq M$. A restriction $N$ of $M$ is also called a {\em submodel} of $M$.
\begin{definition}\label{bisimulation}
    Let $\modelAndTuple \in \classS$ 
    and $\model[\prime] = \modelTuple[\prime] \in \classS$
    be epistemic models. 
    A non-empty relation $\bisimulation \subseteq \states \times \states[\prime]$
    is a {\em bisimulation} if and only if for every 
    $(\stateS, \stateS[\prime]) \in \bisimulation$,
    $\atomP \in \atoms$, and
    $\agentA \in \agents$ 
    the conditions {\bf atoms-$\atomP$}, {\bf forth-$\agentA$} and {\bf back-$\agentA$} hold.
\begin{itemize}
\item    {\bf atoms-$\atomP$}: 
    $\stateS \in \valuation(\atomP)$ if and only if $\stateS[\prime] \in \valuation[\prime](\atomP)$.

\item    {\bf forth-$\agentA$}: 
    For every $\stateT \accessibility{\agentA} \stateS$ 
    there exists $\stateT[\prime] \accessibility[\prime]{\agentA} \stateS[\prime]$
    such that $(\stateT, \stateT[\prime]) \in \bisimulation$.

\item    {\bf back-$\agentA$}: 
    For every $\stateT[\prime] \accessibility[\prime]{\agentA} \stateS[\prime]$
    there exists $\stateT \accessibility{\agentA} \stateS$ 
    such that $(\stateT, \stateT[\prime]) \in \bisimulation$.
\end{itemize}
If $(\stateS, \stateS[\prime]) \in \bisimulation$ then we call
    $\pointedModel{\stateS}$ and $\pointedModel[\prime]{\stateS[\prime]}$
    {\em bisimilar} and write 
    $\pointedModel{\stateS} \bisimilar \pointedModel[\prime]{\stateS[\prime]}$ or (to indicate the relation) $\bisimulation: \pointedModel{\stateS} \bisimilar \pointedModel[\prime]{\stateS[\prime]}$. If for all $s \in S$ there is an $s' \in S'$ such that $\pointedModel{\stateS} \bisimilar \pointedModel[\prime]{\stateS[\prime]}$, and for all $s' \in S'$ there is an $s \in S$ such that $\pointedModel{\stateS} \bisimilar \pointedModel[\prime]{\stateS[\prime]}$, we write $M \bisimilar M'$.
\end{definition}
We note that the union of two bisimulations is a bisimulation, and that there is a maximal bisimulation between the states of an epistemic model, which is an equivalence relation, see \cite{blackburnetal:2001} for such standard notions. A model is {\em bisimulation minimal} iff for any $s,t \in S$ with $s \neq t$, $M_s$ is not bisimilar to $M_t$.

We will also require the notions of {\em restricted bisimulation} (restricted to a set of atoms $Q \subseteq P$) and {\em bounded bisimulation} (bounded to a depth $n \in \Naturals$).   
$\atomsQ$-Bisimulations are intended to preserve modal formulas that contain only atoms from $\atomsQ$, whereas $n$-bisimulations are intended to preserve the truth of formulas $\phi$ with wherein stacks of epistemic operators have maximal depth $n$ (this notion will be defined later).

\begin{definition}\label{q-bisimulation}
    Let $\model, \model[\prime] \in \classS$ be epistemic models and 
    let $\atomsQ \subseteq \atoms$ be a set of propositional atoms.
    A non-empty relation $\bisimulation \subseteq \states \times \states[\prime]$
    is a {\em $\atomsQ$-bisimulation} if and only if for every 
    $(\stateS, \stateS[\prime]) \in \bisimulation$ and $\agentA \in \agents$, {\bf forth-$\agentA$} and {\bf back-$\agentA$} hold, whereas {\bf atoms-$\atomP$} is only required to hold for all $\atomP \in \atomsQ$. If $(\stateS, \stateS[\prime]) \in \bisimulation$ then we call
    $\pointedModel{\stateS}$ and $\pointedModel[\prime]{\stateS[\prime]}$
    {\em $\atomsQ$-bisimilar} and write 
    $\pointedModel{\stateS} \bisimilar[\atomsQ] \pointedModel[\prime]{\stateS[\prime]}$.
\end{definition}
%
The notion of $n$-bisimulation, for $n \in \Naturals$, is given by defining a set of relations $\bisimulation^0 \supseteq \dots \supseteq \bisimulation^n$. 
\begin{definition}\label{n-bisimulation}
Let $\model, \model[\prime] \in \classS$ be epistemic models, and $n \in \Naturals$. A non-empty relation $\bisimulation^0 \subseteq \states \times \states[\prime]$
    is a {\em $0$-bisimulation} if and only if for every 
    $(\stateS, \stateS[\prime]) \in \bisimulation^0$ and for every $p \in P$
    \begin{itemize}
\item    {\bf atoms-$\atomP$}: 
    $\stateS \in \valuation(\atomP)$ if and only if $\stateS[\prime] \in \valuation[\prime](\atomP)$.
\end{itemize}
A non-empty relation $\bisimulation^{n+1} \subseteq \states \times \states[\prime]$
    is an {\em $(n+1)$-bisimulation} if and only if for every 
    $(\stateS, \stateS[\prime]) \in \bisimulation^{n+1}$, for all $\atomP\in\atoms$, and
 for every 
    $\agentA \in \agents$, there is an $n$-bisimulation $\bisimulation^n \supseteq \bisimulation^{n+1}$ such that:
\begin{itemize}
\item    {\bf $(n+1)$-forth-$\agentA$}: 
    For every $\stateT \accessibility{\agentA} \stateS$ 
    there exists $\stateT[\prime] \accessibility[\prime]{\agentA} \stateS[\prime]$
    such that $(t,t') \in \bisimulation^n$;

\item    {\bf $(n+1)$-back-$\agentA$}: 
    For every $\stateT[\prime] \accessibility[\prime]{\agentA} \stateS[\prime]$
    there exists $\stateT \accessibility{\agentA} \stateS$ 
    such that $(t,t') \in \bisimulation^n$.
\end{itemize}
%
If $(\stateS, \stateS[\prime]) \in \bisimulation^n$ for an $n$-bisimulation $\bisimulation^n$, then we call
    $\pointedModel{\stateS}$ and $\pointedModel[\prime]{\stateS[\prime]}$
    {\em $n$-bisimilar} and write 
    $\pointedModel{\stateS} \simeq^n \pointedModel[\prime]{\stateS[\prime]}$.
\end{definition}

\subsection{Syntax and semantics of public announcement logics}

We now define the syntax and semantics of {\em epistemic logic} \logicS, {\em public announcement logic} \logicPal, and {\em arbitrary public announcement logic} \logicApal.

\begin{definition}\label{apal-syntax}
    The {\em language of arbitrary public announcement logic \langApal{}} is the set of {\em formulas} generated by the following rule, where $\atomP \in \atoms$ and $\agentA \in \agents$. Typical members of \langApal{}\ are denoted by lower case Greek letters $\phi$, $\psi$, etc., possibly primed.
    $$
        \phi ::=
            \atomP \mid
            \neg \phi \mid
            (\phi \land \phi) \mid
            \necessary[\agentA] \phi \mid
            \announceA{\phi} \phi \mid
            \allpas \phi
    $$
\end{definition}
We will follow the usual rules for omission of parentheses. 
We use all of the standard abbreviations for propositional logic, and additionally the abbreviations 
$\possible[\agentA] \phi ::= \neg \necessary[\agentA] \neg \phi$,
$\announceE{\phi} \psi ::= \neg \announceA{\phi} \neg \psi$, and
$\somepas \phi ::= \neg \allpas \neg \phi$.
We also consider the language of {\em public announcement logic}, \langPal{}, consisting of \langApal{} without the $\allpas$ operator, the language of {\em epistemic logic}, \langMl{}, consisting of \langPal{} without $\announceA{\cdot}$ operators, and the language of {\em propositional logic}, \langPl, without any modalities. A formula in \langMl{} is an {\em epistemic formula}, and a formula in \langPl\ is a {\em Boolean}. The {\em epistemic depth} of a formula in \langApal{} counts the number of stacked $K_a$ operators (while ignoring the $\Box$ operators), i.e., $d(K_a\phi)=d(\phi)+1$, and $d(p)=0$, $d(\Box\phi)=d(\neg\phi)=d(\phi)$, $d(\phi\et\psi) = \max \{d(\phi),d(\psi)\}$, $d([\phi]\psi) = d(\phi)+d(\psi)$. We write $v(\phi)$ for the set of propositional variables occurring in $\phi$, where $v(\atomP) = \{\atomP\}$, $v(\necessary[\agentA] \phi) = v(\allpas\phi) = v(\neg\phi) = v(\phi)$, and $v([\phi]\psi) = v(\phi\et\psi) = v(\phi) \union v(\psi)$. 


\begin{definition}\label{apal-semantics}
The binary {\em satisfaction} relation $\models$ between pointed epistemic models and $\langApal$ formulas is defined as follows by induction on formula structure. Let $\modelAndTuple \in \classS$ be an epistemic model. Then:
  \[  \begin{array}{lll}
        \pointedModel{\stateS} \entails \atomP &\text{ iff }& \stateS \in \valuation(\atomP)\\
        \pointedModel{\stateS} \entails \neg \phi &\text{ iff }& \pointedModel{\stateS} \nentails \phi\\
        \pointedModel{\stateS} \entails \phi \land \psi &\text{ iff }& \pointedModel{\stateS} \entails \phi \text{ and } \pointedModel{\stateS} \entails \psi\\
        \pointedModel{\stateS} \entails \necessary[\agentA] \phi &\text{ iff }& \text{for every } \stateT \sim_a \stateS : \pointedModel{\stateT} \entails \phi\\
        \pointedModel{\stateS} \entails \announceA{\phi} \psi &\text{ iff }& \text{if } \pointedModel{\stateS} \entails \phi \text{ then } (M|\phi)_s \entails \psi\\
        \pointedModel{\stateS} \entails \allpas \phi &\text{ iff }& \text{for every } \psi \in \langMl : \pointedModel{\stateS} \entails \announceA{\psi} \phi
    \end{array}\]
    where $M|\phi = M|\interpretation[\model]{\phi}$ with $\interpretation[\model]{\phi} = \{\stateS \in \states \mid \pointedModel{\stateS} \entails \phi\}$.
\end{definition}
When $M_{s}\models\phi$, we say that $\phi$ is {\em true} in $M_s$ (or in state $s$ of $M$), or that $M_s$ {\em satisfies} $\phi$. In the semantics of $\allpas$, a $\psi$ such that $\pointedModel{\stateS} \entails \announceA{\psi} \phi$ is called a {\em witness} of the {\em quantifier} $\Box$.

A model restriction $M|\phi$ to a formula $\phi$ restricts the domain of $M$ to those states
where $\phi$ is true. This is the basis of the semantics of public announcements. We note that $\phi$ may no longer be true in that model restriction. A typical counterexample is the Moore sentence $p \et \neg K_a p$: whenever true, after its announcement it is false. The restriction $M|\phi$ is also called the {\em result} of the announcement of $\phi$ in $M$.

Whenever $M_s \models \phi$ for all $s \in S$, we write $M \models \phi$ ($\phi$ is {\em valid on $M$}), and when $M \models \phi$ for all $M$ of class \classS, we write $\classS \models \phi$ and we say that $\phi$ is  {\em valid}. 
Formula $\phi\in\langApal$ is {\em satisfiable} if there is an epistemic model $M_s$ such that $M_s \models \phi$.

Let $M_s$ and $M'_{s'}$ be given. If for all $\phi \in \langApal$, $M_s \models \phi$ if and only if $M'_{s'} \models \phi$, then $M_s$ and $M'_{s'}$ are {\em modally equivalent}, for which we write $M_s \equiv_\mathit{apal} M'_{s'}$. For modal equivalence for formulas up to modal depth $n$ we write $M_s \equiv^n_\mathit{apal} M'_{s'}$, and for modal equivalence for formulas in the language restricted to atoms in $Q \subseteq P$ we write $M_s \equiv^Q_\mathit{apal} M'_{s'}$.

Public announcement logic \logicPal{} and epistemic logic \logicS{}\ have the same semantics as \logicApal{} but defined on the languages \langPal{} and \langMl, respectively. The notation used for modal equivalence in \langMl\ is $\equiv_\mathit{el}$ (we do not need similar notation for \langPal{}, as every formula in \langPal{} is equivalent to a formula in \langMl{} \cite{plaza:1989}, see also the next subsection on expressivity); for the same up to modal depth $n$ it is $\equiv^n_\mathit{el}$, and in the language restricted to atoms in $Q \subseteq P$ it is $\equiv^Q_\mathit{el}$.

We continue with elementary results on the relation between bisimulation and modal equivalence.

\begin{lemma}[{\cite{HennessyM85}}]\label{bisimulation-preserves}
    Let $\pointedModel{\stateS}, \pointedModel[\prime]{\stateS[\prime]} \in \classS$ be epistemic models. Then $\pointedModel{\stateS} \bisimilar \pointedModel[\prime]{\stateS[\prime]}$ implies $\pointedModel{\stateS} \equiv_\mathit{el} \pointedModel[\prime]{\stateS[\prime]}$.
\end{lemma}

\begin{lemma}[{\cite{HennessyM85}}]\label{bisimulation-hennessy-milner}
    Let $M_s, M'_{s'} \in \classS$ be image-finite epistemic models
    (each state has finitely many accessible states). Then $\pointedModel{\stateS} 
 \equiv_\mathit{el} \pointedModel[\prime]{\stateS[\prime]}$ implies 
 $\pointedModel{\stateS} \simeq \pointedModel[\prime]{\stateS[\prime]}$.
\end{lemma}
These are well-known results. 
We observe that Lemma \ref{bisimulation-preserves} can be generalised to the languages \langPal{} and \langApal{} (i.e., to modal equivalence of pointed epistemic models in the respective logics), as public announcements and arbitrary public announcements are bisimulation invariant operations. The latter was shown in \cite{agotnesetal.jal:2010} for the logic {\em GAL}, but the proof also applies to \logicApal; see also the similar proof for \logicPapal{} in Lemma \ref{bisimulation-preserves-papal}, later.

Analogous results to 
Lemma~\ref{bisimulation-preserves} \weg{and Lemma~\ref{bisimulation-hennessy-milner}} apply
to $\atomsQ$-bisimulations when we restrict the language of epistemic formulas 
to propositional atoms in $\atomsQ$, and analogous results also apply to $n$-bisimulations. 
\begin{lemma}[{\cite{dagostinoetal:2000,French06}}]\label{bisimulation-preserves3}
    Let $\pointedModel{\stateS}, \pointedModel[\prime]{\stateS[\prime]} \in \classS$ be epistemic models and let $Q \subseteq P$. Then $\pointedModel{\stateS} \simeq^Q \pointedModel[\prime]{\stateS[\prime]}$ implies $\pointedModel{\stateS} \equiv_\mathit{el}^Q \pointedModel[\prime]{\stateS[\prime]}$.
\end{lemma}
\begin{lemma}[\mbox{\cite[Prop.\ 2.31]{blackburnetal:2001}}]\label{bisimulation-preserves2}
    Let $\pointedModel{\stateS}, \pointedModel[\prime]{\stateS[\prime]} \in \classS$ be epistemic models and let $n \in \Naturals$. Then $\pointedModel{\stateS} \simeq^n \pointedModel[\prime]{\stateS[\prime]}$ implies $\pointedModel{\stateS} \equiv_\mathit{el}^n \pointedModel[\prime]{\stateS[\prime]}$.
\end{lemma}
Again, both generalise to the language \langPal. {\em However, they do not generalise to the language \langApal.} This is because in the restricted logical language the arbitrary announcement still quantifies over all propositional variables and not only over those in $\atomsQ$, and, respectively, because the arbitrary announcement quantifies over formulas of arbitrarily large epistemic depth, and not only over formulas of at most the epistemic depth of the formula bound by the arbitrary announcement. We will get back to this after presenting the expressivity results for public announcement logics, in the next section.

\medskip

A common epistemic model in our contribution is the $a$-$b$-chain. We therefore introduce it in this section, as well as results on distinguishing formulas for $a$-$b$-chains.

Consider the epistemic model $M = (S,\sim,V)$ for two agents $a,b$ and a set of atoms $P$ (often a singleton $P = \{p\}$) such that $S$ is a subset of the integers $\mathbb Z$, $\sim_a$ is the symmetric and reflexive closure of $S^2 \inter \{ (2n,2n+1) \mid n \in \mathbb Z \}$, $\sim_b$ is the symmetric and reflexive closure of $S^2 \inter \{ (2n,2n-1) \mid n \in \mathbb Z \}$, such that between any two states in the domain $S$ a path of  $a$-links and $b$-links exists (in other words, such that $M$ is {\em connected}), and without any requirement on the valuation. As the $a$-links and $b$-links between states alternate in the model, such a model is called an {\em $a$-$b$-chain}, or simply a {\em chain}. The names of the states are arbitrary; any isomorphic model will also be called a chain. A chain is finite iff the domain $S$ is finite. A finite chain has a largest and a smallest element (with respect to $\mathbb Z$) of the domain. These are called the {\em ends} or the {\em edges} of the chain. Observe that an edge is a singleton $\sim_a$-class or $\sim_b$-class, and that all other equivalence classes consist of two states. A chain with only a largest or smallest element has only one edge. Such a one-edged chain is isomorphic to $\mathbb N$. A {\em prefix} of a one-edged $a$-$b$-chain is a submodel that is an $a$-$b$-chain and that contains that edge.

\medskip

We now introduce the {\em distinguishing formula}. Given a logical language $\lang$ and a semantics, such as the above for $\langApal$, and given a model $M = (S,\sim,V)$ and a subset $T \subseteq S$, a {\em distinguishing formula} for $T$ is some $\delta \in \lang$ such that $M_t \models \delta$ for all $t \in T$ and $M_t \not\models\delta$ for all $t \not\in T$.

It is well-known that all subsets of a finite (bisimulation minimal) epistemic model are distinguishable in the language $\langMl$ of epistemic logic (see \cite{jfak.odds:1998} or the more recent \cite{hvdetal.jlc:2014} discussing it; an older source in a slightly different setting is \cite{browneetal:1987}). 
\begin{lemma}\label{lem:distinguishingFormula}
Let $M = (S, \sim, V)$ be a (bisimulation minimal) finite $a$-$b$-chain and let $B \subseteq S$. Then $B$ has a distinguishing formula.
\end{lemma}
Similarly, if the edge of a one-edged infinite $a$-$b$-chain has a distinguishing formula, then all finite subsets of that chain can be distinguished. In order to prove this we first define: $L_{ab}^0 = L_{ba}^0 := \epsilon$, and for $n \geq 0$, $L_{ab}^{2n+1} \delta_0 := L_bL_{ab}^{2n}$,  $L_{ab}^{2n+2} := L_aL_{ab}^{2n+1}$, $L_{ba}^{2n+1} \delta_0 := L_aL_{ba}^{2n}$,  $L_{ba}^{2n+2} := L_bL_{ba}^{2n+1}$. Informally, $L_{ab}^n$ is a stack of $n$ alternating  $L_a$ and $L_b$ operators of which the last one, if any, is $L_b$, whereas $L_{ba}^n$ is a stack of $n$ alternating $L_a$ and $L_b$ operators of which the last one, if any, is $L_a$. Note that for any formula $\phi$, $L_{ab}^0\phi = L_{ba}^0\phi = \phi$. Similarly to $L_{ab}^n$ and $L_{ba}^n$, we define  $K_{ab}^n$ and $K_{ba}^n$.
\begin{lemma} \label{lemma.dist}
Let $M = (S, \sim, V)$ be a one-edged infinite $a$-$b$-chain such that the edge has a distinguishing formula and let $B \subseteq S$ be finite. Then $B$ has a distinguishing formula.
\end{lemma}
\begin{proof}
Without loss of generality we assume that $S = \mathbb N$ (so that $B \subseteq \mathbb N$ is a finite set of natural numbers), that the edge is state $0$, and that $0 \sim_a 1$. Let $\delta_0\in\langMl$ be the assumed distinguishing formula of edge $0$. In other words, $M_0 \models \delta_0$ and for all $i>0$, $M_i \not \models \delta_0$.\footnote{We did not require that $M$ is bisimulation minimal in the formulation of the lemma. The minimality follows from the existence of $\delta_0$. Further note that $M|B$ need not be an $a$-$b$-chain, it may be disconnected.} Obviously, $M_n \models L_{ba}^n \delta_0$. However, also, for all states $i \leq n$, $M_i \models L_{ba}^n \delta_0$, as all states are $a$-accessible and $b$-accessible to themselves. Now let for $n > 0$, $\delta_n := L_{ba}^n \delta_0 \et \neg L_{ba}^{n-1} \delta_0$. From $M_i \models L_{ba}^i \delta_0$ for all $i \leq n$ and $M_j \models \neg L_{ba}^{j-1} \delta_0$ for all $j > n$ it follows that $M_n \models \delta_n$ and that $M_k \not \models \delta_n$ for any $k \neq n$. Therefore $\delta_n$ is a distinguishing formula for state $n \in \mathbb N$, and thus the distinguishing formula $\delta_B$ for $B$ is $\Vel_{i \in B} \delta_i$. 
\end{proof}
We will use this result frequently in subsequent proofs. 

\subsection{Expressivity of public announcement logics} \label{sec.expal}

Given logical languages $\lang$ and $\lang'$, and a class of models in which $\lang$ and $\lang'$ are both interpreted (employing a satisfaction relation $\models$ resp.\ $\models'$), we say that \emph{$\lang$ is at least as expressive as $\lang'$}, if every formula in $\lang'$ is equivalent to a formula in $\lang$ (where `\emph{$\phi'\in\lang'$ is equivalent to $\phi\in\lang$}' means: for all $M_s$, $M_s \models' \phi'$ if and only if $M_s \models \phi$). If $\lang$ is not at least as expressive as $\lang'$ and $\lang'$ is not at least as expressive as $\lang$, then $\lang$ is \emph{incomparable} to $\lang'$ ($\lang$ and $\lang'$ are  incomparable). If $\lang$ is at least as expressive as $\lang'$, and $\lang'$ is at least as expressive as $\lang$, then $\lang$ is \emph{as expressive as} $\lang'$ ($\lang$ and $\lang'$ are equally expressive). Finally, if $\lang$ is at least as expressive as $\lang'$ but $\lang'$ is not at least as expressive as $\lang$, then $\lang$ is {\em more expressive} than $\lang'$. So, `more' means `strictly more'. The combination of a language with a semantics given a class of models determines a logic. In this work we only consider model class $\mathcal{S}5$. Also, in this work the clause of the satisfaction relation for a modality is the same for all languages containing that modality, so that it suffices only to employ $\models$. ``Given logic $L_1$ with language $\lang_1$ interpreted on model class $X_1$ by way of satisfaction relation $\models_1$, and logic $L_2$ with language $\lang_2$ interpreted on model class $X_2$ by way of satisfaction relation $\models_2$, $\lang_1$ is more expressive than $\lang_2$,'' therefore becomes ``given language $\lang$, model class $X$ and satisfaction relation $\models$, logic $L_1$ with language $\lang_1 \subseteq \lang$, and logic $L_2$ with language $\lang_2 \subseteq \lang$, $\lang_1$ is more expressive than $\lang_2$.'' We therefore abbreviate the latter by ``$L_1$ is more expressive than $L_2$,'' and similarly for other expressivity terminology.

The following expressivity results are shown by Plaza~\cite{plaza:1989} (Proposition~\ref{prop.ex1}), and by Balbiani {\em et al.}~\cite{balbianietal:2008} (Propositions~\ref{apal-single-expressiveness} and \ref{apal-expressiveness}). We give the proof of Proposition~\ref{apal-expressiveness} in detail, including an alternative proof (that is not known from the literature), as we will use these methods later when obtaining additional expressivity results for positive arbitrary public announcement logic.

\begin{lemma} \label{prop.ex1}
    \logicPal{} is as expressive as \logicS{} (for single or multiple agents).
\end{lemma}

\begin{proposition}\label{apal-single-expressiveness}
    \logicApal{} is as expressive as \logicPal{} for a single agent.
\end{proposition}

\begin{proposition}\label{apal-expressiveness}
    \logicApal{} is (strictly) more expressive than \logicPal{} for multiple agents.
\end{proposition}

\begin{proof}
   \begin{figure}
        \centering
        \begin{tikzpicture}[>=stealth',shorten >=1pt,auto,node distance=5em,thick]

          \node (s) {\underline{$\{\atomP\}$}};
          \node (t) [right of=s] {$\{\}$};

          \path[every node/.style={font=\sffamily\small}]
            (s) edge node {$\agentA$} (t);
          \node (m) [left of=s, node distance=3em] {$\model_s$:};
        \end{tikzpicture}
 \quad \quad
        \begin{tikzpicture}[>=stealth',shorten >=1pt,auto,node distance=5em,thick]

          \node (s) {\underline{$\{\atomP, \atomQ\}$}};
          \node (t) [right of=s] {$\{\}$};
          \node (s') [below of=s] {$\{\atomP, \atomQ\}$};
          \node (t') [below of=t] {$\{\atomQ\}$};

          \path[every node/.style={font=\sffamily\small}]
            (s) edge node {$\agentA$} (t)
            (s') edge [swap] node {$\agentA$} (t')
            (s) edge [swap] node {$\agentB$} (s')
            (t) edge node {$\agentB$} (t');
           \node (m) [left of=s', node distance=3em] {$\model'_{s'}$:};
       \end{tikzpicture}
\quad \quad
        \begin{tikzpicture}[>=stealth',shorten >=1pt,auto,node distance=5em,thick]

          \node (s) {\underline{$\{\atomP, \atomQ\}$}};
          \node (t) [right of=s] {\color{white} $\{\}$};
          \node (s') [below of=s] {$\{\atomP, \atomQ\}$};
          \node (t') [below of=t] {$\{\atomQ\}$};

           \draw (s') edge [swap] node {$\agentA$} (t');
           \draw (s) edge [swap] node {$\agentB$} (s');
           \node (m) [left of=s', node distance=4em] {$(\model'|q)_{s'}$:};
       \end{tikzpicture}
        \caption{Models used in the proof of Proposition~\ref{apal-expressiveness}. The actual states are underlined. We will always assume reflexive and symmetric closure of accessibility relations.}\label{expressivity-s5i-1}
    \end{figure}

    Suppose that arbitrary public announcement logic is as expressive as public announcement logic in \classS{} for more than one agent.
    We note that public announcement logic is also as expressive as epistemic logic \logicS.
    Consider the formula $\somepas (\knows[\agentA] \atomP \land \neg \knows[\agentB] \knows[\agentA] \atomP)$. Then there exists a formula $\phi \in \langMl$ 
 that is equivalent to $\somepas (\knows[\agentA] \atomP \land \neg \knows[\agentB] \knows[\agentA] \atomP)$. There will be some propositional variable $q$ not occurring in $\phi$. Consider \classS{} models $\model$ and $\model[\prime]$ as in Figure~\ref{expressivity-s5i-1}; let the underlined states be called $s$ and $s'$, respectively. We note that $\pointedModel{\stateS} \simeq^\atomP \pointedModel[\prime]{\stateS[\prime]}$
    and, as $\atomQ$ does not appear in $\phi$, 
    then $\pointedModel{\stateS} \entails \phi$ 
    if and only if $\pointedModel[\prime]{\stateS[\prime]} \entails \phi$.
However $\pointedModel{\stateS} \nentails \somepas (\knows[\agentA] \atomP \land \neg \knows[\agentB] \knows[\agentA] \atomP)$, whereas $\pointedModel[\prime]{\stateS[\prime]} \entails \Dia (\knows[\agentA] \atomP \land \neg \knows[\agentB] \knows[\agentA] \atomP)$ because $\pointedModel[\prime]{\stateS[\prime]} \entails \announceE{\atomQ} (\knows[\agentA] \atomP \land \neg \knows[\agentB] \knows[\agentA] \atomP)$. This is a contradiction.
\end{proof}

Another proof of larger expressivity does not use that $\Dia$ quantifies over {\em arbitrarily many propositional variables} but that $\Dia$ quantifies over {\em formulas of  arbitrarily large epistemic depth}. It is due to Barteld Kooi. It is relevant to mention this alternative proof here, because we will use a similar technique in Section \ref{sec.expressivity} on the expressivity of \logicPapal. Note that all models used in this proof are $a$-$b$-chains.

\begin{proof}
Suppose that arbitrary public announcement logic is as expressive as public announcement logic in \classS{} for more than one agent. Then (again) there exists a formula $\phi \in \langMl$ 
 that is equivalent to $\somepas (\knows[\agentA] \atomP \land \neg \knows[\agentB] \knows[\agentA] \atomP)$. Let $d(\phi)$ be the epistemic depth of this formula. Now consider (see Figure~\ref{fig.exp2}) model $N_t$. We can see it as some sort of infinite $\mathcal{S}5$ unwinding of model $M_s$: $N_t$ is bisimilar to $M_s$. A bisimulation between $M_s$ and $N_t$ links all the $p$-states in $N$ to the single $p$-state in $M$, and all the $\neg p$-states in $N$ to the single $\neg p$-state in $M$. So, in particular, this bisimulation contains pair $(t,s)$. Now consider a model that is like $N_t$, but cut off at the right-hand side, as in model $N'_{t'}$ in Figure~\ref{fig.exp2}, where the cut-off is beyond the epistemic depth of $\phi$: let $j > d(\phi)$ be such that the length of the $a$-$b$-path from the root $t'$ to the edge is $j$ and let that rightmost point be called $v$. State $v$ is the unique state satisfying $K_a p$. Because of Lemma \ref{lemma.dist} we now can uniquely identify all finite subsets of $N'_{t'}$. \weg{For example, the next to rightmost point, let it be $u$, is the unique world where $L_b K_a p \et \neg K_a p$. So the set $\{u,v\}$ is distinguished from its complement by the formula $K_a p \vel (L_b K_a p \et \neg K_a p)$ (and, of course, by the simpler description $L_b K_a p$). And so on.} Therefore, there is an announcement $\psi$ such that $(N'|\psi)_{t'}$ is the final depicted model, where we note that, ignoring the value of $q$, it is the same as the model $(M'|q)_{s'}$ in Figure \ref{expressivity-s5i-1}. Announcement $\psi$ is the formula $(L_{ba}^j K_a p \et \neg L_{ba}^{j-1} K_a p) \vel (L_{ba}^{j+1} K_a p \et \neg L_{ba}^j K_a p) \vel (L_{ba}^{j+2} K_a p \et \neg L_{ba}^{j+1} K_a p)$. (We refer to the proof of Lemma \ref{lemma.dist} for definition of $L_{ab}^n$ and $L_{ba}^n$.) This formula can be simplified to $L_{ba}^{j+2} K_a p \et \neg L_{ba}^{j-1} K_a p$. From $(N'|\psi)_{t'} \models \knows[\agentA] \atomP \land \neg \knows[\agentB] \knows[\agentA] \atomP$ follows $N'_{t'} \models \dia{\psi} (\knows[\agentA] \atomP \land \neg \knows[\agentB] \knows[\agentA] \atomP)$ and thus  $N'_{t'} \models \Dia (\knows[\agentA] \atomP \land \neg \knows[\agentB] \knows[\agentA] \atomP)$. However, as before, $M_s \not\models \Dia (\knows[\agentA] \atomP \land \neg \knows[\agentB] \knows[\agentA] \atomP)$, and therefore, as $M_s \simeq N_t$, also $N_t \not\models \Dia (\knows[\agentA] \atomP \land \neg \knows[\agentB] \knows[\agentA] \atomP)$. On the other hand, ($M_s$ and) $N_t$ and $N'_{t'}$ have the same value for $\phi$, as the difference between the two models is beyond the epistemic depth of $\phi$: 

As $j > d(\phi)$, up to depth $d(\phi)$ the models $N_t$ and $N'_{t'}$ are isomorphic and therefore bisimilar, i.e., $N_t \simeq^{d(\phi)} N'_{t'}$. Now applying Lemma \ref{bisimulation-preserves2} we obtain $N_t \equiv^{d(\phi)}_\mathit{el} N'_{t'}$, and therefore in particular $N_t \models \phi$ iff $N'_{t'} \models \phi$. Again we have a contradiction.
\end{proof}

   \begin{figure}
        \centering

        \begin{tikzpicture}[>=stealth',shorten >=1pt,auto,node distance=5em,thick]
          \node (s) {\underline{$\{\atomP\}$}};
          \node (sr) [right of=s] {$\{\}$};
          \node (srr) [right of=sr] {\color{white} $\{\}$};
          \node (sl) [left of=s] {\color{white} $\{p\}$};
          \node (sll) [left of=sl] {\color{white} $\{\}$};
          \node (slll) [left of=sll] {\color{white} $\{\}$};
          \node (st) [right of=srr] {\color{white} $\{p\}$};
         \node (stt) [right of=st] {\color{white} $\{p\}$};
            \draw  (s) edge node {$\agentA$} (sr);
          \node (m) [left of=slll, node distance=1em] {$\model_s$:};
        \end{tikzpicture}

 \begin{tikzpicture}[>=stealth',shorten >=1pt,auto,node distance=5em,thick]

          \node (s) {\underline{$\{\atomP\}$}};
          \node (sr) [right of=s] {$\{\}$};
          \node (srr) [right of=sr] {$\{\}$};
          \node (sl) [left of=s] {$\{p\}$};
          \node (sll) [left of=sl] {$\{\}$};
          \node (slll) [left of=sll] {\color{white} $\{\}$};
          \node (st) [right of=srr] {\color{white} $\{p\}$};
         \node (stt) [right of=st] {\color{white} $\{p\}$};
         \draw[dashed] (slll) edge (sll);
         \draw  (sll) edge node {$\agentA$} (sl);
           \draw  (sl) edge node {$\agentB$} (s);
           \draw  (s) edge node {$\agentA$} (sr);
           \draw  (sr) edge node {$\agentB$} (srr);
           \draw[dashed]  (srr) edge (st);
            \node (m) [left of=slll, node distance=1em] {$N_t$:};
       \end{tikzpicture}

 \begin{tikzpicture}[>=stealth',shorten >=1pt,auto,node distance=5em,thick]

          \node (s) {\underline{$\{\atomP\}$}};
          \node (sr) [right of=s] {$\{\}$};
          \node (srr) [right of=sr] {$\{\}$};
          \node (sl) [left of=s] {$\{p\}$};
          \node (sll) [left of=sl] {$\{\}$};
          \node (slll) [left of=sll] {\color{white} $\{\}$};
          \node (st) [right of=srr] {$\{p\}$};
         \node (stt) [right of=st] {$\{p\}$};
         \draw[dashed] (slll) edge (sll);
         \draw  (sll) edge node {$\agentA$} (sl);
           \draw  (sl) edge node {$\agentB$} (s);
           \draw  (s) edge node {$\agentA$} (sr);
           \draw  (sr) edge node {$\agentB$} (srr);
           \draw[dashed]  (srr) edge node {$> d(\phi)$} (st);
           \draw  (st) edge node {$\agentB$} (stt);
           \node (m) [left of=slll, node distance=1em] {$N'_{t'}$:};
       \end{tikzpicture}

 \begin{tikzpicture}[>=stealth',shorten >=1pt,auto,node distance=5em,thick]

          \node (s) {\underline{$\{\atomP\}$}};
          \node (sr) [right of=s] {\color{white} $\{\}$};
          \node (srr) [right of=sr] {\color{white} $\{\}$};
          \node (sl) [left of=s] {$\{p\}$};
          \node (sll) [left of=sl] {$\{\}$};
          \node (slll) [left of=sll] {\color{white} $\{\}$};
          \node (st) [right of=srr] {\color{white} $\{p\}$};
         \node (stt) [right of=st] {\color{white} $\{p\}$};
         \draw  (sll) edge node {$\agentA$} (sl);
           \draw  (sl) edge node {$\agentB$} (s);
           \node (m) [left of=slll, node distance=0.05em] {$(N'|\psi)_{t'}$:};
       \end{tikzpicture}

\caption{More models used in the proof of Proposition~\ref{apal-expressiveness}.}\label{fig.exp2}
\end{figure}

These different proofs to establish larger expressivity illustrate a important difference between \logicApal\ and other public announcement logics. Let us first introduce additional notation. Let $M_s \equiv^\atomsQ_\mathit{apal} M'_{s'}$ (where $Q \subseteq P$) mean that for all $\phi\in\langApal$ with atoms restricted to $\atomsQ$, $M_s \models \phi$ if and only if $M'_{s'} \models \phi$, and let $M_s \equiv^n_\mathit{apal} M'_{s'}$ (where $n \in \Naturals$) mean that for all $\phi\in\langApal$ with $d(\phi) \leq n$, $M_s \models \phi$ if and only if $M'_{s'} \models \phi$. 

Although $M_s \simeq M'_{s'}$ implies $M_s \equiv_\mathit{apal} M'_{s'}$, we do not have that $M_s \simeq^Q M'_{s'}$ (always) implies $M_s \equiv^Q_\mathit{apal} M'_{s'}$ and we also do not have that $M_s \simeq^n M'_{s'}$ (always) implies $M_s \equiv^n_\mathit{apal} M'_{s'}$.

The models used in Figure \ref{expressivity-s5i-1} provide typical (counter)examples. We note that \begin{quote} $M_s \simeq^p M'_{s'}$ but $M_s \not\equiv^p_\mathit{apal} M'_{s'}$, \end{quote} because $M_s \models \Dia (K_a p \et \neg K_b K_a p)$ whereas $M'_{s'} \not\models \Dia (K_a p \et \neg K_b K_a p)$. In the language \langApal{} restricted to $p$, the arbitrary announcement modalities are still interpreted over {\em all} atoms $\atoms$, so they quantify not only over  $\langMl$ formulas only containing atom $p$ but also over $\langMl$ formulas possibly containing atom $q$ as well.

Similarly, we note that, as shown in the alternative proof for Prop.~\ref{apal-expressiveness}, \begin{quote} $N_t \simeq^2 N'_{t'}$ but $N_t \not\equiv^2_\mathit{apal} N'_{t'}$. \end{quote} This is because on the one hand $N_t \models \Dia (K_a p \et \neg K_b K_a p)$, whereas on the other hand $N'_{t'} \not\models \Dia (K_a p \et \neg K_b K_a p)$, as $N'_{t'} \not\models \dia{\psi} (K_a p \et \neg K_b K_a p)$ for some $\psi \in \langMl$ with $d(\psi)>2$.

However, this does not rule out that bounded (to some $n$) bisimilarity implies bounded modal equivalence {\em in a particular model}. This will be used in an expressivity proof comparing \logicApal\ and \logicPapal, later (Theorem \ref{expressivity-s52} on page \pageref{expressivity-s52}).

\subsection{Positive formulas}

The {\em positive formulas} are the universal fragment of epistemic logic. They play an important role in our work, also in relation to the structural notion of {\em refinement}, that is therefore only defined in this section.

\begin{definition}\label{positive-formula}
    The {\em language of positive formulas \langMlPlus{}} is defined inductively as:
    $$
        \phi ::=
            \atomP \mid
            \neg \atomP \mid
            (\phi \land \phi) \mid
            (\phi \lor \phi) \mid
            \necessary[\agentA] \phi
    $$
    where $\atomP \in \atoms$ and $\agentA \in \agents$.
\end{definition}
We note that $\langMlPlus$ is a fragment of $\langMl$.

\begin{lemma}\label{public-announcement-preserves}
    Positive formulas are {\em preserved} under public announcements: \\ For all $\phi \in \langMlPlus$, $\psi \in \langMl$: $M_s \models \phi$ implies $M_s \models [\psi]\phi$.
\end{lemma}

\begin{corollary}
    Positive formulas are {\em successful} as public announcements: \\ For all $\phi \in \langMlPlus$: $M_s \models \phi$ implies $M_s \models [\phi]\phi$.
\end{corollary}

\begin{corollary}
    Positive formulas are {\em idempotent} as public announcements: \\  For all $\phi \in \langMlPlus$, $\psi \in \langMl$: $M_s \models [\phi]\psi$ implies $M_s \models [\phi][\phi]\psi$.
\end{corollary}
These results were shown by van Ditmarsch and Kooi~\cite[Prop.\ 8]{hvdetal.synthese:2006} for an extended fragment also containing the inductive clause $[\neg\phi]\phi$.  In their work positive formulas are called {\em preserved formulas} instead. The results go back to van Benthem \cite{jfak.lonely:2006}.

A {\em refinement} is a relation that is a generalisation of bisimulation, and that only requires the {\bf atoms} and {\bf back} condition to hold. Refinements (in this form) were introduced in \cite{bozzellietal.inf:2014}.

\begin{definition}\label{refinement}
    Let $\model, \model[\prime] \in \classS$ be epistemic models.
    A non-empty relation $\refinement \subseteq \states \times \states[\prime]$
    is a {\em refinement} if and only if for every
    $(\stateS, \stateS[\prime]) \in \bisimulation$,
    $\atomP \in \atoms$, and
    $\agentA \in \agents$,
    the conditions {\bf atoms-$\atomP$} and {\bf back-$\agentA$} hold.
    If $(\stateS, \stateS[\prime]) \in \refinement$ then we call $\pointedModel[\prime]{\stateS[\prime]}$ a {\em refinement} of $\pointedModel{\stateS}$
    and call $\pointedModel{\stateS}$ a {\em simulation} of $\pointedModel[\prime]{\stateS[\prime]}$.
    We write $\pointedModel[\prime]{\stateS[\prime]} \refines \pointedModel{\stateS}$
    or equivalently $\pointedModel{\stateS} \simulates \pointedModel[\prime]{\stateS[\prime]}$.
\end{definition}
There are other notions of refinement. For example, in \cite{blackburnetal:2001}, simulation is defined with an inclusion requirement for atoms: for each pair $(s,s')$ in the relation, $s \in V(p)$ \emph{implies} $s' \in V'(p)$; instead of full correspondence {\bf atoms}: $s\in V(p)$ \emph{if and only if} $s' \in V'(p)$. The dual of that notion of simulation leads to a different notion of refinement.

\begin{lemma}\label{refinement-preorder}
    The relation $\simulates$ is a preorder on epistemic models.
\end{lemma}

\begin{lemma}\label{refinement-preserves}
    Let $\pointedModel{\stateS}, \pointedModel[\prime]{\stateS[\prime]} \in \classS$ be epistemic models
    such that $\pointedModel{\stateS} \simulates \pointedModel[\prime]{\stateS[\prime]}$ and
    let $\phi \in \langMlPlus$ be a positive formula.
    If $\pointedModel{\stateS} \entails \phi$ then $\pointedModel[\prime]{\stateS[\prime]} \entails \phi$.
\end{lemma}
These results were shown by Bozzelli {\em et al.} in \cite[Prop.~2 \& 8]{bozzellietal.inf:2014}. We note that the union of two refinements is a refinement and that there is a maximal refinement between epistemic models (this is shown just as for bisimulation).

The relation between positive formulas and refinement is intricate. One important (and unreported) result is as follows. It follows from similar reasoning to Lemma~\ref{bisimulation-hennessy-milner}, but in view of its novelty and because we will refer to it later, we give the proof.

\begin{lemma}\label{refinement-hennessy-milner}
    Let $\model, \model[\prime] \in \classS$ be image-finite epistemic models
    and let $\refinement \subseteq \states \times \states[\prime]$ 
    be the relation such that $(\stateS, \stateS[\prime]) \in \refinement$ 
    if and only if for every $\phi \in \langMlPlus$:
    if $\pointedModel{\stateS} \entails \phi$ then $\pointedModel[\prime]{\stateS[\prime]} \entails \phi$.
    If $\refinement$ is non-empty, then $\refinement$ is a refinement.
\end{lemma}

\begin{proof}

Let $(s,s') \in \refinement$.

The clause {\bf atoms} is satisfied because $\pointedModel{\stateS} \entails p$ implies $\pointedModel[\prime]{\stateS[\prime]} \entails p$, and  also $\pointedModel{\stateS} \entails \neg p$ implies $\pointedModel[\prime]{\stateS[\prime]} \entails \neg p$. 

Let us now consider {\bf back}, and suppose this clause is not satisfied. Then there is a $t'$ with $s' \sim_a t'$ (i.e., an $a$-successor $t'$ of $s'$) such that none of the finite $a$-successors $t_1,\dots,t_n$ of $s$ are in the relation $\refinement$ with $t'$, i.e., $(t_1,t') \not\in \refinement$,\dots,$(t_n,t') \not\in\refinement$. Therefore, using the definition of $\refinement$, for each $t_i$, where $i = 1,\dots,n$, there is a $\phi_i \in \langMlPlus$ such that $M_{t_i} \models \phi_i$ but $M'_{t'} \not \models \phi_i$. Therefore $M_s \models K_a (\phi_1 \vel \dots \vel \phi_n)$, whereas $M_{s'} \not\models K_a (\phi_1 \vel \dots \vel \phi_n)$. Observe that  $K_a (\phi_1 \vel \dots \vel \phi_n)$ is a positive formula (the positive formulas are closed under disjunction and under $K_a$). This contradicts our assumption that $(s,s') \in \refinement$.
\end{proof}

\section{Positive arbitrary public announcement logic}\label{syntax-semantics}

In this section we give the syntax and semantics of {\em positive arbitrary public announcement logic} \logicPapal{}, and we
provide some semantic results about the properties of positive announcements
and arbitrary positive announcement operators.

\begin{definition}\label{papal-syntax}
    The {\em language of positive arbitrary public announcement logic \langPapal{}} is defined inductively as:
    $$
        \phi ::=
            \atomP \mid
            \neg \phi \mid
            (\phi \land \phi) \mid
            \necessary[\agentA] \phi \mid
            \announceA{\phi} \phi \mid
            \allppas \phi
    $$
    where $\atomP \in \atoms$ and $\agentA \in \agents$. 
\end{definition}
We use the abbreviation 
$\someppas \phi ::= \neg \allppas \neg \phi$. The {\em epistemic depth} and the {\em set of variables} of a formula are defined as before.

\begin{definition}\label{papal-semantics}
    Let $\modelAndTuple \in \classS$ be an epistemic model.
    The interpretation of $\phi \in \langPapal$
    is defined inductively as in Def.~\ref{apal-semantics}, but with the following clause for positive arbitrary announcement:
    \begin{eqnarray*}
        \pointedModel{\stateS} \entails \allppas \phi &\text{ iff }& \text{ for every } \psi \in \langMlPlus : \pointedModel{\stateS} \entails \announceA{\psi} \phi
    \end{eqnarray*}
\end{definition}
So, the difference with the semantics for the arbitrary announcement in \logicApal{} is the part `for every $\psi \in \langMlPlus$' instead of `for every $\psi \in \langMl$'. Validity, satisfiability and modal equivalence (notation $\equiv_\mathit{apal+}$) are defined as before. 

An important observation is the partial correspondence between the results of
positive announcements and model restrictions that are closed under
refinements, a notion that we will define now.

\begin{definition}\label{closed-under-refinements}
    Let $\modelAndTuple \in \classS$ be an epistemic model and 
    let $\statesT \subseteq \states$ be a non-empty set of states.
    We say that {\em $\statesT$ is closed under refinements in $\model$}
    if and only if for every $\stateS, \stateT \in \states$ such that $\pointedModel{\stateS} \simulates \pointedModel{\stateT}$:
    if $\stateS \in \statesT$ then $\stateT \in \statesT$.
    We say that the model restriction {\em $\restrict{\model}{\statesT}$ is closed under refinements} 
    if and only if $\statesT$ is closed under refinements in $\model$.
\end{definition}

\begin{lemma}\label{positive-announcements-refinement-closed}
    The result of any positive announcement is closed under refinements.
\end{lemma}

\begin{proof}
    Let $\modelAndTuple \in \classS$ be an epistemic model and
    let $\phi \in \langMlPlus$. We have to show that a non-empty $M|\phi$ (i.e., $M| \interpretation[\model]{\phi}$) is closed under refinements. 
    Suppose that $\stateS, \stateT \in \states$ 
    such that $\stateS \in \interpretation[\model]{\phi}$ 
    and $\pointedModel{\stateS} \simulates \pointedModel{\stateT}$.
    Then $\pointedModel{\stateS} \entails \phi$.
    As $\pointedModel{\stateS} \simulates \pointedModel{\stateT}$ and $\phi \in \langMlPlus$
    then by Lemma~\ref{refinement-preserves} we have $\pointedModel{\stateT} \entails \phi$.
    So $\stateT \in \interpretation[\model]{\phi}$ 
    and therefore $\restrict{\model}{\phi}$ is closed under refinements.
\end{proof}

\begin{lemma}\label{refinement-closed-positive-announcements}
    On finite models, a model restriction that is closed under refinements is the result of a positive announcement.
\end{lemma}

\begin{proof}
    Let $\modelAndTuple \in \classS$ be an epistemic model and
    let $\statesT \subseteq \states$ be a non-empty set of states such that
    $\restrict{\model}{\statesT}$ is closed under refinements.
    Then for every $\stateS \in \statesT$ and $\stateT \in \states \setminus \statesT$
    we have that $\pointedModel{\stateS} \not\simulates \pointedModel{\stateT}$. As $\model$ is image-finite it then follows from Lemma \ref{refinement-hennessy-milner} that for every $\stateS \in \statesT$ and $\stateT \in \states \setminus \statesT$ there exists $\phi_{\stateS,\stateT} \in \langMlPlus$
    such that $\pointedModel{\stateS} \entails \phi_{\stateS,\stateT}$
    but $\pointedModel{\stateT} \nentails \phi_{\stateS,\stateT}$.
    Let $\phi = \bigvee_{\stateS \in \statesT} \bigwedge_{\stateT \in \states \setminus \statesT} \phi_{\stateS, \stateT}$.
    Then $\phi \in \langMlPlus$; 
    for every $\stateS \in \statesT$: $\pointedModel{\stateS} \entails \phi$; 
    and for every $\stateT \in \states \setminus \statesT$: $\pointedModel{\stateT} \nentails \phi$.
    So $\interpretation[\model]{\phi} = \statesT$ and 
    therefore $\restrict{\model}{\statesT}$ is the result of a positive announcement.
\end{proof}

In contrast to public announcements, a sequence of positive announcements
cannot generally be expressed as a single positive announcement.

\begin{proposition}
    Arbitrary positive announcements are not composable in \classS{},
    i.e., it is not the case that $\classS \entails \someppas \someppas \phi \implies \someppas \phi$ for all $\phi \in \langPapal$.
\end{proposition}

\begin{figure}
\scalebox{.68}{
        \begin{tikzpicture}[>=stealth',shorten >=1pt,auto,node distance=7em,thick]

          \node[label=above left:$\stateS$] (s) {\underline{$\{\atomP\}$}};
          \node[label=above right:{$\stateS[\prime]$}] (x) [right of=s] {$\{\atomP, \atomQ\}$};
          \node[label=above left:$\stateT$] (t) [below of=s] {$\{\atomP, \atomQ\}$};
          \node[label=above left:$\stateU$] (u) [below of=t] {$\{\atomP\}$};
          \node[label=below left:$\stateV$] (v) [below of=u] {$\{\atomP, \atomQ\}$};
          \node[label=above left:$\stateW$] (w) [left of=u] {$\{\}$};
          \node[label=above right:{$\stateT[\prime]$}] (t') [below of=x] {$\{\atomP, \atomQ\}$};
          \node[label=above right:{$\stateU[\prime]$}] (u') [below of=t'] {$\{\atomP\}$};
          \node[label=above right:{$\stateV[\prime]$}] (v') [below of=u'] {$\{\atomP, \atomQ\}$};
          \node () [right of=u'] {};

          \path[every node/.style={font=\sffamily\small}]
            (s) edge node [left] {$\agentA$} (t)
            (x) edge [bend right] node {$\agentA$} (t)
            (s) edge node {$\agentA, \agentB$} (x)
            (x) edge [->,dotted, bend left] node {$\simulates$} (t')
            (t) edge node [left] {$\agentB$} (u)
            (u) edge node [left] {$\agentA$} (v)
            (u) edge node [above] {$\agentA$} (w)
            (w) edge node [below left] {$\agentA$} (v)
            (s) edge [bend left] node [left] {$\agentA$} (t')
            (x) edge node {$\agentA$} (t')
            (t') edge node [right] {$\agentB$} (u')
            (u') edge node [right] {$\agentA$} (v')
            (t) edge node {$\agentA$} (t')
            (t) edge [->,dotted,bend right] node {$\simulates$} (t')
            (u) edge [->,dotted] node {$\simulates$} (u')
            (v) edge [->,dotted] node {$\simulates$} (v')
            (s) edge [->,dotted,bend left=115,min distance=5cm,swap] node {$\simulates$} (u')
            (t') edge [->,dotted,bend left=90,min distance=2.5cm,swap] node {$\simulates$} (v');
        \end{tikzpicture}

\hspace{-1.5cm} 
        \begin{tikzpicture}[>=stealth',shorten >=1pt,auto,node distance=7em,thick]

          \node[label=above left:$\stateS$] (s) {\underline{$\{\atomP\}$}};
          \node[label=above right:{$\stateS[\prime]$}] (x) [right of=s] {$\{\atomP, \atomQ\}$};
          \node[label=above left:$\stateT$] (t) [below of=s] {$\{\atomP, \atomQ\}$};
          \node[label=below left:{\color{white}$\stateV$}] (v) [below of=u] {\color{white} $\{\atomP, \atomQ\}$};
          \node[label=above right:{$\stateT[\prime]$}] (t') [below of=x] {$\{\atomP, \atomQ\}$};
          \node[label=above right:{$\stateU[\prime]$}] (u') [below of=t'] {$\{\atomP\}$};
          \node[label=above right:{$\stateV[\prime]$}] (v') [below of=u'] {$\{\atomP, \atomQ\}$};
          \node () [right of=u'] {};

          \path[every node/.style={font=\sffamily\small}]
            (s) edge node [left] {$\agentA$} (t)
            (x) edge [bend right] node {$\agentA$} (t)
            (s) edge node {$\agentA, \agentB$} (x)
            (s) edge [bend left] node [left] {$\agentA$} (t')
            (x) edge node {$\agentA$} (t')
            (t') edge node [right] {$\agentB$} (u')
            (u') edge node [right] {$\agentA$} (v')
            (t) edge node {$\agentA$} (t');
        \end{tikzpicture}
\hspace{0cm}
        \begin{tikzpicture}[>=stealth',shorten >=1pt,auto,node distance=7em,thick]

          \node[label=above left:$\stateS$] (s) {\underline{$\{\atomP\}$}};
          \node[label=above right:{\color{white} $\stateS[\prime]$}] (x) [right of=s] {\color{white} $\{\atomP, \atomQ\}$};
          \node[label=above left:$\stateT$] (t) [below of=s] {$\{\atomP, \atomQ\}$};
          \node[label=below left:{\color{white}$\stateV$}] (v) [below of=u] {\color{white} $\{\atomP, \atomQ\}$};
          \node[label=above right:{\color{white}$\stateT[\prime]$}] (t') [below of=x] {\color{white}$\{\atomP, \atomQ\}$};
          \node[label=above right:{$\stateU[\prime]$}] (u') [below of=t'] {$\{\atomP\}$};
          \node[label=above right:{$\stateV[\prime]$}] (v') [below of=u'] {$\{\atomP, \atomQ\}$};
          \node () [right of=u'] {};

          \path[every node/.style={font=\sffamily\small}]
            (s) edge node [left] {$\agentA$} (t)
            (u') edge node [right] {$\agentA$} (v');
        \end{tikzpicture}
}
         \caption{Counterexample for the composability of positive announcements. Left: initial model, with explicit refinement relation. Middle: after announcing $\knows_\agentA \atomP$. Right: after subsequently announcing $\neg \atomQ \vel \knows_\agentB \atomQ$.}\label{composability-s5-1}
   \end{figure}

\begin{proof}
    We construct a counter-example. Consider model $\modelAndTuple$ in Figure \ref{composability-s5-1}, where 
    $\states = \{\stateS, \stateT, \stateU, \stateV, \stateW, \stateS[\prime], \stateT[\prime], \stateU[\prime], \stateV[\prime]\}$,
    $\stateS \accessibility{\agentA} \stateT \accessibility{\agentA} \stateS[\prime] \accessibility{\agentA} \stateT[\prime]$,
    $\stateU \accessibility{\agentA} \stateV \accessibility{\agentA} \stateW$,
    $\stateU[\prime] \accessibility{\agentA} \stateV[\prime]$,
    $\stateS \accessibility{\agentB} \stateS[\prime]$,
    $\stateT \accessibility{\agentB} \stateU$,
    $\stateT[\prime] \accessibility{\agentB} \stateU[\prime]$,
    $\valuation(\atomP) = \{\stateS, \stateT, \stateU, \stateV, \stateS[\prime], \stateT[\prime], \stateU[\prime], \stateV[\prime]\}$, and
    $\valuation(\atomQ) = \{\stateT, \stateV, \stateS[\prime], \stateT[\prime], \stateV[\prime]\}$.
%
%

    We claim that $\pointedModel{\stateS} \entails \someppas \someppas (\suspects[\agentA] \atomQ \land \knows[\agentA] (\knows[\agentB] \atomQ \lor \knows[\agentB] \neg \atomQ))$
    but $\pointedModel{\stateS} \nentails \someppas (\suspects[\agentA] \atomQ \land \knows[\agentA] (\knows[\agentB] \atomQ \lor \knows[\agentB] \neg \atomQ))$.

    We note that $(\restrict{\restrict{\model}{\knows[\agentA] \atomP}}{(\lnot \atomQ \lor \knows[\agentB] \atomQ)})_{\stateS} \entails \suspects[\agentA] \atomQ \land \knows[\agentA] (\knows[\agentB] \atomQ \lor \knows[\agentB] \neg \atomQ)$ (see  Figure~\ref{composability-s5-1}) 
    and so $\pointedModel{\stateS} \entails \someppas \someppas (\suspects[\agentA] \atomQ \land \knows[\agentA] (\knows[\agentB] \atomQ \lor \knows[\agentB] \neg \atomQ))$.

    Let $\refinement = \{(\stateX, \stateX) \mid \stateX \in \states\} \cup \{(\stateT, \stateT[\prime]), (\stateU, \stateU[\prime]), (\stateV, \stateV[\prime]), (\stateS[\prime], \stateT[\prime]), (\stateT[\prime], \stateV[\prime]), (\stateS, \stateU[\prime])\}$. 
    We note that $\refinement$ is a refinement.

As $\pointedModel{\stateS[\prime]} \simulates \pointedModel{\stateT[\prime]}$, 
     $\pointedModel{\stateT} \simulates \pointedModel{\stateT[\prime]}$, and $\pointedModel{\stateT[\prime]} \simulates \pointedModel{\stateT[\prime]}$, 
    then by Lemma~\ref{positive-announcements-refinement-closed}
    any positive announcement that is true in $M_s$ and that preserves $\stateS[\prime]$ or $\stateT$ (or $\stateT[\prime]$) will also preserve $\stateT[\prime]$
    so any positive announcement after which $\suspects[\agentA] \atomQ$ is true will preserve $\stateT[\prime]$.

    Also, as $\pointedModel{\stateS} \simulates \pointedModel{\stateU[\prime]}$
    then by Lemma~\ref{positive-announcements-refinement-closed}
    any positive announcement that is true in $M_s$ will preserve $\stateU[\prime]$. 

Therefore, any positive announcement that is true in $M_s$ and preserves one of the $a$-accessible states $t$, $s'$ and $t'$ (such that $L_a q$ is true after the announcement), will also preserve $t'$, and also $u'$ that is $b$-accessible from $t'$, such that $\neg (K_b q \vel K_b \neg q)$ is true in $t'$ after the announcement, and therefore $L_a \neg (K_b q \vel K_b \neg q)$ true in $s$.

So if after any positive announcement in $s$ $\suspects[\agentA] \atomQ$ is true, then $\neg \knows[\agentA] (\knows[\agentB] \atomQ \lor \knows[\agentB] \neg \atomQ)$ is also true.
    Therefore $\pointedModel{\stateS} \entails \allppas (\suspects[\agentA] \atomQ \implies \neg \knows[\agentA] (\knows[\agentB] \atomQ \lor \knows[\agentB] \neg \atomQ))$
    and so $\pointedModel{\stateS} \nentails \someppas (\suspects[\agentA] \atomQ \land \knows[\agentA] (\knows[\agentB] \atomQ \lor \knows[\agentB] \neg \atomQ))$.
\end{proof}

However, other validities involving the quantifier are the same for \logicApal{} and for \logicPapal{}. We list a few.

\begin{lemma}[\mbox{\cite[Lemmas 3.1 \& 3.9]{balbianietal:2008}}] \label{shit1}
Let $\phi,\psi\in\langPapal$. Then: 
\begin{enumerate}
\item $\classS \models \allppas (\phi\et\psi) \eq (\allppas\phi\et\allppas\psi)$; 
\item $\classS \models \phi$ implies $\classS \models \allppas\phi$; 
\item $\classS \models \allppas \phi \imp\phi$;
\item $\classS \models K_a \allppas \phi \imp \allppas K_a \phi$.
\end{enumerate}
\end{lemma}
\begin{proof}
The proofs are exactly as in \cite{balbianietal:2008}. The first two directly follow from the semantics of $\allppas$. For the third, note that  $\classS \models \allppas \phi \imp [\top]\phi$ and $\classS \models [\top] \phi \eq \phi$. 

For the last, suppose that $M_s\models K_a \allppas \phi$, that $M_s \models \psi$ for $\psi\in\langMlPlus$, and that $t$ is in the domain of $M|\psi$ such that $s \sim_a t$. We have to show that $(M|\psi)_t \models \phi$. As $s\sim_a t$ also holds in $M$, from the assumption $M_s \models K_a \allppas \phi$ it follows that $M_t \models \allppas \phi$. As $t$ is in the domain of $M|\psi$, $M_t \models \psi$. From $M_t \models \allppas \phi$ and $M_t \models \psi$ follows $(M|\psi)_t \models \phi$, as required.
\end{proof}

We now proceed to prove that, like the \logicApal{} quantifier, the \logicPapal{} quantifier also satisfies the Church-Rosser and McKinsey properties. As our proof of Church-Rosser is very different from that in \cite{balbianietal:2008}, we give the proofs of these properties and also the proofs of the lemmas on which they depend in detail.

Given a model $M = (S,\sim,V)$, if for all $p \in Q \subseteq P$, $V(p) = \emptyset$ or $V(p) = S$, we say that the valuation of the atoms in $Q$ is \emph{constant} on $M$.

\begin{lemma}[\mbox{\cite[Lemma 3.2]{balbianietal:2008}}]\label{shit2}
Let $\phi\in\langPapal$ and let $M \in \classS$ be a model on which the valuation of atoms in $v(\phi)$ is constant. Then $M \models \phi$ or $M \models \neg\phi$.
\end{lemma}
\begin{proof}
Let $p \in v(\phi)$. Let $\psi(\top/p)$ be the substitution of all occurrences of $p$ in $\phi$ by $\top$. As the valuation of $p$ is constant, then if $p$ is true on $M$ we have that $M \models \phi \eq \phi(\top/p)$. Similarly, if $p$ is false on $M$, then $M \models \phi \eq \phi(\bot/p)$. Let $\phi'$ be the result of successively substituting all $p \in v(\phi)$ by $\top$ or $\bot$ in this way. Clearly $M \models \phi \eq \phi'$. Using the $\classS$ validities $K_a \top \eq \top$, $K_a \bot \eq \bot$, using $\classS \models \allppas \top \eq \top$ and $\classS \models \allppas \bot \eq \bot$, and using propositional properties of combining $\top$ and $\bot$, we obtain that $\classS \models \phi' \eq \top$ or $\classS \models \phi' \eq \top$. Thus we also have that $M \models \phi \eq \top$ or $M \models \phi \eq \top$, in other words, that $M \models \phi$ or $M \models \neg\phi$.
\end{proof}

\begin{lemma}[\mbox{\cite[Lemma 3.3]{balbianietal:2008}}]\label{shit3}
Let $\phi\in\langPapal$ and let $M = (S,\sim,V) \in \classS$ be a model on which the valuation of atoms in $v(\phi)$ is constant. Then $M \models \phi\imp\allppas\phi$.
\end{lemma}
\begin{proof}
Let $s \in S$ and $M_s \models \phi$. Now consider $\psi \in \langMlPlus$ such that $M_s \models \psi$. Note that the valuation of atoms in $v(\phi)$ on $M|\psi$ is also constant, and that it is the same as the valuation of atoms on $M$. Consider the disjoint union $M' = M + M|\psi$ of $M$ and $M|\psi$ (it is the model defined by the disjoint union of the respective domains, accessibility relations, and valuations). Like $M$ and $M|\psi$, $M'$ is a model on which the valuation of $v(\phi)$ is constant. Therefore, from Lemma \ref{shit2}, either $M' \models \phi$ or $M' \models \neg\phi$. The second implies that $M \models \neg\phi$ which contradicts $M_s \models \phi$. Therefore, $M' \models \phi$. From $M' \models \phi$ it follows that $M|\psi \models \phi$, so that in particular $(M|\psi)_s \models \phi$. As $\psi$ was arbitrary, we now have shown that: for all $\psi \in \langMlPlus$, if $M_s \models \psi$ then $(M|\psi)_s \models \phi$. By the semantics of public announcement this is equivalent to: for all $\psi \in \langMlPlus$, $M_s \models [\psi] \phi$. By the semantics of $\allppas$ this is equivalent to $M_s \models \allppas \phi$. From that and the assumption it follows that $M_s \models \phi \imp \allppas \phi$, and as $s$ was arbitrary, we thus have shown that $M \models \phi \imp \allppas\phi$, as required.
\end{proof}

Given a model $M_s$ and a (finite) set of propositional variables $\atomsQ$, we write $\delta_\atomsQ^s$ for the conjunction of literals expressing the values of the atoms from $\atomsQ$ in $s$. This is the so-called {\em characteristic formula of the (restricted) valuation in state $s$}. Note that $M_s \models \delta_Q^s$.

\begin{proposition}\label{church-rosser-s5}
    Arbitrary positive announcements have the Church-Rosser property in \classS{},
    i.e. $\classS \entails \someppas \allppas \phi \implies \allppas \someppas \phi$ for all $\phi \in \langPapal$.
\end{proposition}

\begin{proof}
Let $M_s\models \someppas \allppas \phi$. Then there is $\psi\in\langMlPlus$ such that $M_s \models \dia{\psi}\allppas \phi$, i.e., $M_s \models \psi$ and $(M|\psi)_s \models \allppas \phi$. In particular, $(M|\psi)_s \models [\delta_{v(\phi)}^s] \phi$. Therefore, since we also have that $M_s \models \delta_{v(\phi)}^s$, $(M|\psi|\delta_{v(\phi)}^s)_s \models \phi$. Observe that $M|\psi|\delta_{v(\phi)}^s$ is a model on which the valuation of atoms in $v(\phi)$ is constant.

Now let $\eta \in \langMlPlus$ be arbitrary and such that $M_s \models \eta$. Consider $(M|\eta|\delta_{v(\phi)}^s)_s$. The valuation of the atoms in $v(\phi)$ is also constant on $M|\eta|\delta_{v(\phi)}^s$. From Lemma \ref{shit2} it follows that $M|\eta|\delta_{v(\phi)}^s \models \phi$ or $M|\eta|\delta_{v(\phi)}^s \models \neg \phi$. Similarly to the reasoning in the proof of Lemma \ref{shit3}, the latter contradicts (the above) $(M|\psi|\delta_{v(\phi)}^s)_s \models \phi$, and therefore  $M|\eta|\delta_{v(\phi)}^s \models \phi$. From that follows $(M|\eta)_s \models \dia{\delta_{v(\phi)}^s} \phi$, so that $(M|\eta)_s \models \someppas \phi$. 

We have now shown that for all $\eta\in\langMlPlus$, $M_s \models \eta$ implies $(M|\eta)_s \models \someppas \phi$, which by the semantics of $\allppas$ is equivalent to $M_s \models \allppas \someppas \phi$. 
\end{proof}

\begin{proposition}[\mbox{\cite[Prop.\ 3.4]{balbianietal:2008}}]\label{prop.mckinsey}
    Arbitrary positive announcements have the McKinsey property in \classS{},
    i.e. $\classS \entails \allppas \someppas \phi \implies \someppas \allppas \phi$ for all $\phi \in \langPapal$.
\end{proposition}

\begin{proof}
The proof is different from that of Proposition~\ref{church-rosser-s5}, but the crucial role of the announcement of values of all variables in $\phi$ is the same. 

Suppose $M_s \models \allppas \someppas \phi$. Then (as above) $(M|\delta_{v(\phi)}^s)_s \models \someppas \phi$. As $M|\delta_{v(\phi)}^s$ is a model on which the valuation of the atoms in $v(\phi)$ is constant, we have (Lemma \ref{church-rosser-s5}) that $M|\delta_{v(\phi)}^s \models \phi \imp \allppas \phi$. Also using the dual $M|\delta_{v(\phi)}^s \models \someppas\phi \imp \phi$ of that lemma, we obtain $M|\delta_{v(\phi)}^s \models \someppas\phi \imp \allppas\phi$. From that and $(M|\delta_{v(\phi)}^s)_s \models \someppas \phi$ it follows that $(M|\delta_{v(\phi)}^s)_s \models \allppas \phi$, and we therefore obtain $M_s \models \dia{\delta_{v(\phi)}^s}\allppas \phi$ and also $M_s \models \someppas\allppas \phi$.
\end{proof}

Nowhere in the proofs of Lemma \ref{shit1}, Lemma \ref{shit2}, Proposition~\ref{prop.mckinsey} and Proposition~\ref{church-rosser-s5} is it essential that the announced formulas are positive. A closer comparison with the results in \cite{balbianietal:2008} may therefore be of interest: 

The proofs of Lemma \ref{shit1} and Proposition~\ref{prop.mckinsey} (McKinsey) are virtually identical to, respectively, \cite[Lemma 3.2]{balbianietal:2008} and \cite[Prop.\ 3.4]{balbianietal:2008}. The proof of Lemma \ref{shit2} is more detailed than that of \cite[Lemma 3.3]{balbianietal:2008}, but it seems to amount to the intentions of that more schematic proof. However, the proof of Proposition~\ref{church-rosser-s5} (Church-Rosser) is very different from the proof of \cite[Prop.\ 3.8]{balbianietal:2008}, that is not only more involved but also based on a lemma that was later shown by Kuijer to be incorrect. We have therefore not attributed Proposition~\ref{church-rosser-s5} to \cite{balbianietal:2008}. 

\medskip

Again, we note that analogous results to 
Lemma~\ref{bisimulation-preserves} and Lemma~\ref{bisimulation-hennessy-milner} on bisimulation correspondence also apply to the language $\langPapal$: bisimilarity preserves modal equivalence (\logicPapal{} is bisimulation invariant), and on image-finite models modal equivalence implies bisimilarity. As we also consider some variations, we will give the crucial detail to prove the first.


\begin{lemma}\label{bisimulation-preserves-papal}
Let $\pointedModel{\stateS}, \pointedModel[\prime]{\stateS[\prime]} \in \classS$. Then $\pointedModel{\stateS} \bisimilar \pointedModel[\prime]{\stateS[\prime]}$ implies $\pointedModel{\stateS} \equiv_\mathit{apal+}\pointedModel[\prime]{\stateS[\prime]}$.
\end{lemma}

\begin{proof}
We prove the equivalent proposition: \begin{quote} Let $\phi \in \langPapal$ be a formula. Then for all epistemic models $\pointedModel{\stateS}, \pointedModel[\prime]{\stateS[\prime]} \in \classS$ such that $\pointedModel{\stateS} \bisimilar \pointedModel[\prime]{\stateS[\prime]}$, $\pointedModel{\stateS} \entails \phi$ if and only if $\pointedModel[\prime]{\stateS[\prime]} \entails \phi$.\end{quote} This is a straightforward proof by induction over the complexity of the formula $\phi$  occurring the proposition (just as the proof in~\cite{agotnesetal.jal:2010}), where it is important that this formula is declared {\em before} the two models. The clause for $\allppas \phi$ goes as follows.

\bigskip

\noindent $\pointedModel{\stateS} \entails \allppas\phi$ \\ $\Eq$ \\
For all $\psi\in\langMlPlus$, $\pointedModel{\stateS} \entails [\psi]\phi$ \\ $\Eq$ \\ 
For all $\psi\in\langMlPlus$, $M_s \models \psi$ implies $(M|\psi)_s \entails \phi$ \\ $\Eq$  \hfill (*) \\ 
For all $\psi\in\langMlPlus$, \pmb{$M'_{s'} \models \psi$} implies $(M|\psi)_s \entails \phi$ \\ $\Eq$  \hfill (**) \\
For all $\psi\in\langMlPlus$, $M'_{s'} \models \psi$ implies \pmb{$(M'|\psi)_{s'} \entails \phi$} \\ $\Eq$ \\
For all $\psi\in\langMlPlus$, $M'_{s'} \entails [\psi]\phi$ \\ $\Eq$ \\
$M'_{s'} \entails \allppas\phi$.

\bigskip

Step (*) is justified because $\pointedModel{\stateS} \bisimilar \pointedModel[\prime]{\stateS[\prime]}$ implies $\pointedModel{\stateS} \equiv_\mathit{el} \pointedModel[\prime]{\stateS[\prime]}$ (Lemma \ref{bisimulation-preserves}). Therefore, $M_s \models \psi$ iff $M'_{s'} \models \psi$. 

Step (**) is justified as follows. Given the assumption $\pointedModel{\stateS} \bisimilar \pointedModel[\prime]{\stateS[\prime]}$, let $\bisimulation: M_s \bisimilar M'_{s'}$. Define $\bisimulation'$ as follows: $\bisimulation'(t,t')$ iff ($\bisimulation(t,t')$ and $M_t \models \psi$). From Lemma \ref{bisimulation-preserves} it follows that also $M'_{t'} \models \psi$, so that $\bisimulation'$ is indeed a relation between $M|\psi$ and $M'|\psi$. We now show that $\bisimulation':(M|\psi)_s \bisimilar (M'|\psi)_{s'}$. The clause {\bf atoms}-$p$ is obviously satisfied. Concerning {\bf forth}-$a$, take any pair $(t,t')$ such that $\bisimulation'(t,t')$ and let $u$ in the domain of $M|\psi$ be such that $t \sim_a u$. As $u$ is in the domain of $M|\psi$, $M_u \models \psi$. From $\bisimulation'(t,t')$ follows $\bisimulation(t,t')$. As $t \sim_a u$ in $M|\psi$, also $t \sim_a u$ in $M$. From $\bisimulation(t,t')$, $t \sim_a u$ in $M$, and {\bf forth}-$a$ (for $\bisimulation$) it follows that there is $u'$ in the domain of $M'$ such that $\bisimulation(u,u')$ and $t' \sim_a u'$. From $\bisimulation(u,u')$, $M_u \models \psi$, and Lemma \ref{bisimulation-preserves} it follows that $M'_{u'} \models \psi$, i.e., $u'$ is also in the domain of $M|\psi$. From $\bisimulation(u,u')$, $M_u \models \psi$, and the fact the $u'$ is in the domain of $M|\psi$ it follows that $\bisimulation'(u,u')$, as required. This proves {\bf forth}-$a$. The step {\bf back}-$a$ is shown similarly. Note that in particular $\bisimulation'(s,s')$. This therefore establishes that $(M|\psi)_s \simeq (M'|\psi)_{s'}$.

We now use that the induction hypothesis for $\phi$ not merely holds for epistemic models $M_s$, $ M'_{s'}$, but for any pair of epistemic models $N_t$, $N'_{t'}$, so in particular for $(M|\psi)_s$, $(M'|\psi)_{s'}$. (In the formulation of the proposition to be proved, the formula is declared before the models.) We thus conclude that $(M|\psi)_s \entails \phi$ iff $(M'|\psi)_{s'} \entails \phi$, as required.
\end{proof}

The above lemma is the analogue of bisimulation invariance for epistemic logic (Lemma \ref{bisimulation-preserves}). 
This analogue does not hold for \logicPapal{} when we replace bisimulations with $Q$-bisimulations. 
This is because in the formula $\allppas\phi$, the positive announcements can range over atoms that do not appear in $\phi$. It similarly fails for the logic \logicApal{}; see the discussion at the end of Subsection \ref{sec.expal}.

\medskip

We end this section with a fairly obvious result for compactness, that is similarly obtained for \logicPapal\ as for \logicApal.

A logic with language $\lang$ is {\em compact} if for any $\Phi\subseteq\lang$, if every finite $\Phi'\subseteq\Phi$ is satisfiable, then $\Phi$ is satisfiable. Like \logicApal{}, \logicPapal{} is not compact.

\begin{proposition}
    \logicPapal{} is not compact.
\end{proposition}

This follows from the same reasoning used by Balbiani {\em et al.}~\cite{balbianietal:2008} to show that \logicApal{} is not compact.
Specifically, under the semantics of \logicPapal{} the set of formulas $\{\announceA{\psi} (\knows[\agentA] \atomP \implies \knows[\agentB] \knows[\agentA] \atomP) \mid \psi \in \langMlPlus\} \cup \{\neg \allppas (\knows[\agentA] \atomP \implies \knows[\agentB] \knows[\agentA] \atomP)\}$ is unsatisfiable but every finite subset is satisfiable. (The only difference with the proof in \cite{balbianietal:2008} is that instead of `$\psi \in \langMlPlus$' above it there says `$\psi \in \langMl$'.) Any finite subset is satisfiable because the epistemic depth of such a set $\{\announceA{\psi} (\knows[\agentA] \atomP \implies \knows[\agentB] \knows[\agentA] \atomP) \mid \psi \in \langMlPlus\}$ is bounded (or, alternatively, because there must be an atom $q \in P$ not occurring in such a set). We then proceed fairly similarly as in the proof of Proposition \ref{apal-expressiveness}.

\section{Model checking complexity} \label{sec.modelchecking}

We now address the model checking complexity for \logicPapal. 
In this section we will assume that we are working with a finite fragment of the language 
(so $\agents$ and $\atoms$ are finite sets) and finite models $\modelAndTuple$ (where $\states$ is finite). 
The {\em model checking problem} for \logicPapal, for which we write $\mathbf{MC(\logicPapal)}$, is as follows: 
given a finite pointed model $M_s$ and $\phi\in\langPapal$, determine whether $M_s \models \phi$. 
The model checking problem for \logicPapal{} is PSPACE-complete. 
We adapt the proof given for the PSPACE-complete model checking complexity for $GAL$, 
by {\AA}gotnes {\em et al.}~\cite[Theorems 24 \& 25]{agotnesetal.jal:2010}. 
We note that \logicApal{} model checking is also PSPACE-complete, 
which was shown in \cite[p.~74]{agotnesetal.jal:2010} 
by an even simpler adaptation of the proof for $GAL$ than our adaptation for \logicPapal.




\begin{lemma}\label{maximal-refinement}
    Let $\modelAndTuple, \modelAndTuple[\prime] \in \classS$ be finite epistemic models.
    Given that a refinement from $\model$ to $\model[\prime]$ exists,  
    there is a unique, maximal refinement $\refinement \subseteq \states \times \states[\prime]$
    from $\model$ to $\model[\prime]$ 
    and it is computable in polynomial time.
\end{lemma}

\begin{proof}
    This follows from similar reasoning used to show that the unique, maximal
    bisimulation between two models is computable in polynomial time, defining
    the refinement as a greatest fixed point of a monotone function, however
    relaxing the {\bf forth} condition appropriately.
    Specifically, we define the function 
    $f:\wp(S\times S')\longrightarrow\wp(S\times S')$ by $(s,s')\in f(\mathcal{R})$ if and only if:
    \begin{itemize}
      \item $(s,s')\in\mathcal{R}$;
      \item for all $\atomP\in\atoms$, $s\in V(\atomP)$ if and only if $s'\in V'(\atomP)$;
      \item for all $\agentA\in\agents$, for every $t'\sim'_\agentA s'$, 
        there exists $t\sim_\agentA s$ such that $(t,t')\in\mathcal{R}$;
    \end{itemize}
    It is clear that the function is monotone 
    (i.e., if $\mathcal{R}\subseteq\mathcal{R}'$, then $f(\mathcal{R})\subseteq f(\mathcal{R}')$), 
    and that any non-empty fixed point of this function is a refinement.
    Furthermore, every refinement, $\refinement$ from $\model$ to $\model[\prime]$, is a fixed point of $f$.
    Therefore, the greatest fixed point will be a unique maximal refinement. 
    The function $f$ can be computed in polynomial time, 
    and at most $|S|\cdot|S'|$ iterations will be required to reach a fixed point,
    so the maximal refinement is computable in polynomial time.
\end{proof}

\begin{theorem}\label{papal-model-checking-pspace}
    $\mathbf{MC(\logicPapal)}$ is in PSPACE.
\end{theorem}

\begin{proof} 
    We adapt an alternating polynomial time (APTIME) model-checking algorithm used for $GAL$ \cite[Algorithm 1, p.\ 74]{agotnesetal.jal:2010}.
    The main variation required is that we must be able to test whether a submodel can be defined by a positive announcement.
    From Proposition~\ref{positive-announcements-refinement-closed} and
    Proposition~\ref{refinement-closed-positive-announcements} it follows that, in
    order for a restriction of a finite model to be definable as the result of
    a positive announcement, it must be closed under refinements. 
    From Lemma~\ref{maximal-refinement} we can check that this condition is
    satisfied by first computing in polynomial time the maximal refinement from
    $M$ to itself, $\mathfrak{R}^M$, 
    and then using this refinement to select only refinement preserving restrictions of the model.
  
We now present the algorithm, $sat$, for model-checking in \logicPapal. 
  The algorithm $sat$ takes as input a finite model $M = (S,\sim, V)$, some $s\in S$ 
  and a formula $\varphi\in\langPapal$ that we require to be in what is known as {\em negation normal form}. 
  This means that the formula conforms to the syntax 
  $\varphi::= \atomP\ |\ \neg\atomP\ |\ \varphi\et\varphi\ |\ \varphi\vel\varphi\ 
  |\ \knows_\agentA\varphi\ |\ \suspects_\agentA \varphi\ |\ \announceA{\varphi}\varphi\ |\ \announceE{\varphi}\varphi\ 
  |\ \allppas\varphi\ |\ \someppas\varphi$. 
  It is clear that all formulas are semantically equivalent to a formula in negation normal form. One can easily compute it and its size is linear in the size of the given formula.

  A run of the algorithm halts with either {\em accept} or {\em reject}.
  Each case of the algorithm is either existential or universal, 
  where for an existential case to be accepting, one choice must lead to an accepting case, 
  and for a universal case to be accepting, every choice must lead to an accepting case.
  The algorithm is presented as Algorithm~\ref{modelCheckingAlg}.

\begin{algorithm}[t]
\caption{$sat(M,s,\phi)$}
\label{modelCheckingAlg}
\begin{algorithmic}[]

  \STATE \textbf{from} $M$ \textbf{compute} $\mathfrak{R}^M$;
  \STATE \textbf{case} $\phi$ \textbf{of}
\qquad\STATE{($\cdot$)} $p$: \textbf{if} {$s\in V(p)$} \textbf{then} accept \textbf{else}  reject;
\qquad \STATE{$(\cdot)$} $\neg p$: \textbf{if} $s\in V(p)$ \textbf{then} reject \textbf{else} accept;
\qquad\STATE{$(\forall)$} $\phi_{1}\land\phi_{2}$: choose $\phi'\in\{\phi_{1},\phi_{2}\}$; $sat(M,s,\phi')$;
\qquad\STATE{$(\exists)$} $\phi_{1}\vee\phi_{2}$: choose $\phi'\in\{\phi_{1},\phi_{2}\}$; $sat(M,s,\phi')$;
\qquad\STATE{$(\forall)$} $\knows_a\phi'$: choose $t\sim_a s$; $sat(M,t,\phi')$;
\qquad\STATE{$(\exists)$} $\suspects_a\phi'$: choose $t\sim_a s$; $sat(M,t,\phi')$;
\qquad\STATE{$(\forall)$} $\announceA{\phi_{1}}\phi_{2}$: choose a restriction $M'=(S',\sim',V')$ of $M$;\\ 
  \qquad\qquad\qquad\textbf{if for some} $s'\in S'$, \textbf{not} $sat(M,s',\phi_1)$ \textbf{then} accept\\ 
  \qquad\qquad\qquad\textbf{else if for some} $s'\in S\setminus S'$, $sat(M,s',\phi_1)$ \textbf{then} accept\\ 
  \qquad\qquad\qquad\textbf{else if} $s\notin S'$ \textbf{then} accept \textbf{else} $sat(M',s,\phi_{2})$;
\qquad\STATE{$(\exists)$} $\announceE{\phi_{1}}\phi_{2}$: choose a restriction $M'=(S',\sim',V')$ of $M$;\\ 
  \qquad\qquad\qquad\textbf{if for some} $s'\in S'$, \textbf{not} $sat(M,s',\phi_1)$ \textbf{then} reject\\ 
  \qquad\qquad\qquad\textbf{else if for some} $s'\in S\setminus S'$, $sat(M,s',\phi_1)$ \textbf{then} reject\\ 
  \qquad\qquad\qquad\textbf{else if} $s\notin S'$ \textbf{then} reject \textbf{else} $sat(M',s,\phi_{2})$;
\qquad\STATE{$(\forall)$} $\allppas\phi$: Choose any  
restriction $M'=(S',\sim',V')$ of $M$
  such that\\
  \qquad\qquad\qquad for all $t\in S'$, for all $t'$ where $(t,t')\in\mathfrak{R}^M$, we have $t'\in S'$\\
  \qquad\qquad\qquad\textbf{if} $s\in S'$ \textbf{then} $sat(M',s,\phi)$ \textbf{else} accept;
\qquad\STATE{$(\exists)$} $\someppas\phi$: Choose any
restriction $M'=(S',\sim',V')$ of $M$
  such that\\
  \qquad\qquad\qquad for all $t\in S'$, for all $t'$ where $(t,t')\in\mathfrak{R}^M$, we have $t'\in S'$\\
  \qquad\qquad\qquad\textbf{if} $s\in S'$ \textbf{then} $sat(M',s,\phi)$ \textbf{else} reject;
  \STATE \textbf{end case}
\end{algorithmic}
\end{algorithm}

The proof of correctness follows the inductive argument presented in \cite{agotnesetal.jal:2010}:
we can show that $sat(M,s,\phi)$ accepts if and only if $M_s\models\phi$ by induction over the complexity of $\phi$.
The correctness of the $(\forall)\allppas\phi$ and $(\exists)\someppas\phi$ cases follows directly from 
Proposition~\ref{positive-announcements-refinement-closed},
Proposition~\ref{refinement-closed-positive-announcements} and Lemma~\ref{maximal-refinement}, as mentioned above. 
In particular, the $(\forall)\allppas\phi$ and $(\exists)\someppas\phi$ cases 
in Algorithm~\ref{modelCheckingAlg} can be shown to match the semantic interpretation of 
the $\allppas$ and $\someppas$ operators respectively.
Focusing on the $\someppas$ case, if there is a positive announcement $\alpha$ such that $(\restrict{M}{\alpha})_s\models\phi$, 
then there is some restriction $M' = (S',\sim', V')$ of $M$ such that $s\in S'$, 
$M'_s\models\phi$, and for all $t,t' \in S$ such that $M_t \simulates M_{t'}$, if $t \in S'$ then $t' \in S'$. 
Applying the inductive hypothesis, choosing such a restriction will lead to an accepting run.
Conversely, if there is some restriction satisfying these properties, 
from Lemma~\ref{refinement-closed-positive-announcements}, 
there must be some corresponding positive announcement $\alpha$ that realises this refinement. 
Again, applying the inductive hypothesis we have $(\restrict{M}{\alpha})_s\models\phi$ 
and hence $M_s\models\someppas\alpha$. The case for $\allppas$ is treated in a similar manner. 

Since $sat$ can be implemented in polynomial time, {\bf MC(\logicPapal)} is in APTIME, 
which is equivalent to PSPACE \cite{chandraetal:1981}.
\end{proof}

  \usetikzlibrary{calc}
  \begin{figure}[ht]
  \center
  \begin{tikzpicture}[>=stealth',shorten >=1pt,auto,node distance=7em,thick]
    \node (s) {$\begin{array}{c}(s_1,0)\\x^-_1\end{array}$};
    \node (sr) [right of=s] {$\begin{array}{c}(s_2,0)\\x^-_2\end{array}$};
    \node (srr) [right of=sr] {$\dots$};
    \node (srrr) [right of=srr] {$\begin{array}{c}(s_k,0)\\x^-_k\end{array}$};
    \node (sl) [below of=s] {$\begin{array}{c}(s_1,1)\\x^+_1\end{array}$};
    \node (slr) [right of=sl] {$\begin{array}{c}(s_2,1)\\x^+_2\end{array}$};
    \node (slrr) [right of=slr] {$\dots$};
    \node (slrrr) [right of=slrr] {$\begin{array}{c}(s_k,1)\\x^+_k\end{array}$};
    \node (tmp) at ($(s)!0.5!(sl)$) {}; 
    \node (o) [left of=tmp] {$\begin{array}{c}s\\x_0\end{array}$}; 
      \node (m) [left of=o] {$M^\Psi$:};
      \draw (s) -- (o);
      \draw (sl) -- (o);
      \draw (s) -- (sr);
      \draw (sr) -- (srr);
      \draw (srr) -- (srrr);
      \draw (sl) -- (slr);
      \draw (slr) -- (slrr);
      \draw (slrr) -- (slrrr);
      \draw (s) -- (sl);
      \draw (sr) -- (slr);
      \draw (srrr) -- (slrrr);
       \end{tikzpicture}
    \caption{The model $M^\Psi$ used to encode the quantified Boolean formula $\Psi$. The agent's relation is universal, so it is the transitive closure of the depicted relation.}\label{QBF-encode}
  \end{figure}

\begin{theorem}\label{papal-model-checking-pspace-hard}
    $\mathbf{MC(\logicPapal)}$ is PSPACE-hard.
\end{theorem}

\begin{proof}
    This follows from similar reasoning used to show that
    $\mathbf{MC(GAL)}$ is PSPACE-hard, \cite{agotnesetal.jal:2010}. 
    The basic approach is to show that instances of the
    QBF-SAT problem can be solved through model-checking a \langPapal{} formula
    on an appropriately constructed model. 
    A quantified Boolean formula may be given as 
    $\Psi = Q_1x_1\hdots Q_kx_k\phi(x_1,\hdots,x_k)$, 
    where $Q_i\in \{\forall,\exists\}$, 
    $x_1,\hdots,x_k$ are propositional variables, and
    $\phi(x_1,\hdots,x_k)$ is a Boolean formula. 
    (Following custom in QBF-SAT, the variables are not named $p_1,\dots,p_n$ but $x_1,\dots,x_n$ instead.) 
    The notation $\phi(x_1,\hdots,x_k)$ means that each variable $x_1,\dots,x_k$ binds to all its occurrences, 
    possibly none, in $\phi$. 
    For $1 \leq n \leq k$, we will use the abbreviation 
    $\Psi^n = Q_nx_n\hdots Q_kx_k\phi(x_1,\hdots,x_k)$ 
    to represent the fragment of $\Psi$ where $x_1,\dots,x_{n-1}$ (if $n=1$ none at all) are unquantified.

    The satisfiability problem for quantified Boolean problems (QBF-SAT) 
    is well-known to be the canonical problem for PSPACE-completeness.
    Given any quantified Boolean formula $\Psi$, we can construct a 
    model, $M_s^\Psi$, and a \logicPapal{} formula, $\psi$, 
    such that $M_s^\Psi\models\psi$ if and only if $\Psi$ is satisfiable.
    
    The model $M^\Psi = (S,\sim, V)$ is specified with respect to 
    a set of atoms $\{x^+_i,x_i^- |\ 1\leq i\leq k\}$, an additional auxiliary variable $x_0$, and a single agent.
    The model represents each Boolean variable $x_i$ for $1\leq i\leq k$ by a pair of states $(s_i,0)$ and $(s_i,1)$, so we let
    $S = \{s_1,\hdots,s_k\}\times\{0,1\}\cup\{s\}$, 
    including the $s$ state as designated state from which to evaluate the formula.
    The single agent with a universal relation is unable to distinguish any state, so ${\sim} = S\times S$.
    Finally we have $V(x^+_i) = \{(s_i,1)\}$ and $V(x^-_i) = \{(s_i,0)\}$, and $V(x_0) = \{s\}$. The model is depicted in Figure~\ref{QBF-encode}.

    We are then able to encode the satisfiability of $\Psi = Q_1x_1\dots Q_kx_k\phi(x_1,\dots,x_k)$ inductively. 
    For each $i$ from 1 to $k$ we define the formulas $X_i = L_jx^+_i$, $\overline{X}_i = L_jx^-_i$, $U_i = X_i\land\overline{X}_i$ and $D_i =X_i\leftrightarrow\lnot\overline{X}_i$, for, respectively: $X_i$ is true, $X_i$ is false, $X_i$ is undetermined and $X_i$ is determined. Additionally, $U_0$ represents $L_j x_0$.
    As the base case of the induction we have $f(\phi(x_1,\dots,x_k)) = U_0\land \bigwedge_{i=1}^k D_i\land \phi(X_1,\dots,X_k)$,
    that may be satisfied by some restriction of $M^\Psi$, as in Figure~\ref{QBF-encode2}.
  
  \begin{figure}[ht]
  \center
  \begin{tikzpicture}[>=stealth',shorten >=1pt,auto,node distance=7em,thick]
    \node (s) {$\begin{array}{c}(s_1,0)\\x^-_1\end{array}$};
    \node (sr) [right of=s] {};
    \node (srr) [right of=sr] {};
    \node (srrr) [right of=srr] {};
    \node (sl) [below of=s] {};
    \node (slr) [right of=sl] {$\begin{array}{c}(s_2,1)\\x^+_2\end{array}$};
    \node (slrr) [right of=slr] {$\dots$};
    \node (slrrr) [right of=slrr] {$\begin{array}{c}(s_k,1)\\x^+_k\end{array}$};
      \node (tmp) at ($(s)!0.5!(sl)$) {}; 
    \node (o) [left of=tmp] {$\begin{array}{c}s\\x_0\end{array}$}; 
      \node (m) [left of=o] {\color{white} $M^\Psi$:};
      \draw (s) -- (o);
      \draw (s) -- (slr);
      \draw (slr) -- (slrrr);
       \end{tikzpicture}
    \caption{A restriction of  $M^\Psi$ satisfying $f(\lnot x_1\land x_2\land x_k)$}.\label{QBF-encode2}
  \end{figure}

    We have two inductive cases, for $Q_n = \forall$ and $Q_n = \exists$. 
    If $Q_n = \forall$, then 
    $$f(\Psi^n) = \allppas\left(\left(U_0\land\bigwedge_{i=1}^n D_i\land\bigwedge_{i=n+1}^kU_i\right)\rightarrow f(\Psi^{n+1})\right).$$ 
    If $Q_n = \exists$, then 
    $$f(\Psi^n) = \someppas\left(U_0\land\bigwedge_{i=1}^n D_i\land\bigwedge_{i=n+1}^kU_i\land f(\Psi^{n+1})\right).$$ 
    The Boolean quantifier $Q_i$ is simulated using a $\someppas$ or $\allppas$ operator, 
    appropriately guarded so that it removes precisely one of $(s_i,0)$ and $(s_i,1)$ from the model, 
    and does not affect any of the other states. 
    This is achieved by the subformulas: $\bigwedge_{i=1}^n D_i$, 
    which requires either $(s_i,0)$ or $(s_i,1)$ to remain in the restriction, for $i=1\dots n$;
    and $\bigwedge_{i=n+1}^kU_i$ which requires both $(s_i,0)$ and $(s_i,1)$ to remain in the restriction for $i=n+1\dots k$.
    Since the model is finite, and each state has a unique evaluation, this can always be achieved by a positive public announcement.
    After all the quantifiers have been applied in turn, we able to interpret the Boolean formula $\phi$, by checking which states remain.
    The encoding of this formula and the constructed model are polynomial in the size of $\Psi$,
    so model-checking \logicPapal{}\ is PSPACE-hard. 
    A more extensive discussion of the construction and proof can be found in \cite{agotnesetal.jal:2010}.
\end{proof}


\section{Expressivity} \label{sec.expressivity}

\subsection{The relative expressivity of APAL$^+$ and PAL}

In this section we establish various expressivity results, mainly that (for more than one agent) \logicPapal{} is more expressive than \logicS{} (or \logicPal), which is obvious, and that \logicPapal{} and \logicApal{} are incomparable, which is not obvious. 

\begin{proposition}
    Arbitrary positive announcement logic is as expressive as public announcement logic in \classS{} for a single agent.
\end{proposition}

\begin{proof} We recall that single-agent \logicApal\ is as expressive as \logicS\ \cite[Prop.~3.11 and 3.12]{balbianietal:2008}. The same proof applies to single-agent \logicPapal: it plays no role anywhere in the proof in \cite{balbianietal:2008} whether the announcement witnessing an \logicApal\ quantifier is an epistemic formula or a positive epistemic formula. \end{proof}

\begin{proposition} \label{prop.exp1}
  Arbitrary positive announcement logic is (strictly) more expressive than public announcement logic in \classS{} for more than one agent.
\end{proposition}

\begin{proof}
We refer to the proof of Proposition~\ref{apal-expressiveness}. Observe that the announcement $\atomQ$ used in that proof is a positive formula. Therefore, this also shows that no epistemic formula is equivalent to the \langPapal{} formula $\someppas (K_a p \et \neg K_b K_a p)$. 
\end{proof}

\subsection{APAL$^+$ is not at least as expressive as APAL}

We now consider the relative expressivity of \logicApal{} and \logicPapal{}. In this subsection we show in Theorem \ref{expressivity-s5}, further below, that
    \logicPapal{} is not at least as expressive as \logicApal{} for multiple agents, by the standard method of providing two pointed epistemic models and a formula (in \langApal) such that the models can be distinguished by that formula but cannot be distinguished by any formula in the other language (in \langPapal). The theorem and its proof are preceded by the definition of the respective models and by various lemmas to be used in that proof. The next subsection is devoted to the other direction of expressivity, namely that \logicApal{} is not at least as expressive as \logicPapal{} for multiple agents. From these two results we can then conclude that \logicPapal{} and \logicApal{} are incomparable in expressivity.

  \begin{figure}
        \centering
\scalebox{1}{
           \begin{tikzpicture}[>=stealth',shorten >=1pt,auto,node distance=7em,thick]

              \node (0) [label=above:{$1$}] {$\{\}$};
              \node (1) [right of=0,label=above:{$2$}] {$\{\}$};
              \node (2) [right of=1,label=above:{$3$}] {$\{\}$};
              \node (3) [right of=2] {$\cdots$};
 
              \node (0p) [node distance=5em,below of=0,label=below:{$1'$}] {$\{\}$};
              \node (1p) [right of=0p,label=below:{$2'$}] {$\{\}$};
              \node (2p) [right of=1p,label=below:{$3'$}] {$\{\}$};
              \node (3p) [right of=2p] {$\cdots$};

              \node[label=above:{$0$}] (s) [left of=0] {\underline{$\{\atomP\}$}};
              \node[label=below:{$0'$}] (t) [left of=0p] {$\{\atomP\}$};
             \node (m) [node distance=4em,left of=s] {$M^\omega_0:$};

             \path[every node/.style={font=\sffamily\small}]
                (0) edge node {$\agentB$} (1)
                (1) edge node {$\agentA$} (2)
                (2) edge node {$\agentB$} (3);
             \path[every node/.style={font=\sffamily\small}]
                 (0p) edge node {$\agentB$} (1p)
                (1p) edge node {$\agentA$} (2p)
                (2p) edge node {$\agentB$} (3p);
\draw (s) edge node {$\agentB$} (t);
\draw (s) edge node {$\agentA$} (0);
\draw (t) edge node {$\agentA$} (0p);
            \end{tikzpicture}
}
\bigskip
\bigskip

\scalebox{1}{
           \begin{tikzpicture}[>=stealth',shorten >=1pt,auto,node distance=7em,thick]

              \node (0) [label=above:{$1$}] {$\{\}$};
              \node (1) [right of=0,label=above:{$2$}] {$\{\}$};
              \node (2) [right of=1,label=above:{$3$}] {$\{\}$};
              \node (3) [right of=2] {$\cdots$};
 
              \node (0p) [node distance=5em,below of=0,label=below:{$1'$}] {$\{\}$};
              \node (1p) [right of=0p,label=below:{$2'$}] {$\{\}$};
              \node (2p) [right of=1p] {$\cdots$};
              \node (3p) [right of=2p,label=below:{$m'$}] {$\{\}$};

              \node[label=above:{$0$}] (s) [left of=0] {\underline{$\{\atomP\}$}};
              \node[label=below:{$0'$}] (t) [left of=0p] {$\{\atomP\}$};
              \node (m) [node distance=4em,left of=s] {$M^m_0:$};

             \path[every node/.style={font=\sffamily\small}]
                (0) edge node {$\agentB$} (1)
                (1) edge node {$\agentA$} (2)
                (2) edge node {$\agentB$} (3);
             \path[every node/.style={font=\sffamily\small}]
                 (0p) edge node {$\agentB$} (1p)
                (1p) edge node {$\agentA$} (2p)
                (2p) edge node {$\agentB$} (3p);
\draw (s) edge node {$\agentB$} (t);
\draw (s) edge node {$\agentA$} (0);
\draw (t) edge node {$\agentA$} (0p);
            \end{tikzpicture}
}

\caption{Models $M^\omega_0$ and $M^m_0$ used in the proof of Theorem~\ref{expressivity-s5}.}\label{expressivity-s5i-1mm} 
\end{figure}

Consider the models $M^\omega_0$ and $M^m_0$ in Figure \ref{expressivity-s5i-1mm}. We will use these two epistemic models in the expressivity result of this section. 
They are both $a$-$b$-chains. We first define the base model, $M^\omega$. 
Formally $M^\omega = (S,\sim,V)$ where $S = \Naturals \union \Naturals'$ (where $\Naturals' = \{ i' \mid i \in \Naturals \}$), 
$\sim_a$ is the symmetric and reflexive closure of $\{ (2i,2i+1), (2i',(2i+1)') \mid i \in \Naturals\}$, 
$\sim_b$ is the symmetric and reflexive closure of $\{ (2i+1,2i+2), ((2i+1)',(2i+2)') \mid i \in \Naturals\} \union \{(0,0')\}$, and $V(p) = \{0,0'\}$.

We define an ordering $\preceq$ over $S\cup\{\omega,\omega'\}$ as follows. This use of the symbol $\preceq$ is different from that for models in the refinement relation and is therefore unambiguous. 
$$x\preceq y\textrm{ iff }\left\{
  \begin{array}{l}
    x\in\Naturals\textrm{ and } y=\omega\\
    x,y\in\Naturals\textrm{ and }x\leq y,\\
    x,y\in\Naturals', x=w',y=z',\textrm{ and }w\geq z,\\
    x=\omega'\textrm{ and } y\in\Naturals', \textrm{ or }\\
    x\in\Naturals'\cup\{\omega'\}\textrm{ and } y\in\Naturals\cup\{\omega\}
  \end{array}
  \right.
$$  
For convenience, in this proof we will denote the set $S_x^y = \{z\in S \mid x\preceq z\preceq y\}$ where $x,y\in\Naturals\cup\{\omega,\omega'\}$.
For $m \in\Naturals$, $M^m_0$ will be used as an abbreviation for the model $(M|S_{m'}^\omega)_0$, as depicted in Figure~\ref{expressivity-s5i-1mm}. 
We will show that no formula of \logicPapal{} can distinguish the set of models $\{M^m_0\mid m\in\Naturals\}$ from the set of models $\{M^m_0\mid m\in\Naturals\} \cup \{ M^\omega_0\}$, while the sets are distinguishable in \logicApal{}. The assumption that such a distinguishing \logicPapal{} formula exists is contradictory, as it must then in particular be true in $M^m_0$ for some $m \in \Naturals$ sufficiently large, in which case we can show that it must also be true in $M^\omega_0$. 

\medskip

Models $M^m_0$ and $M^\omega_0$ in Figure~\ref{expressivity-s5i-1mm} are the same except that in $M^m_0$ the lower leg is cut off at the world named $m'$. 
As $m$ is arbitrary, the final indistinguishability link between $(m-1)'$ and $m'$ could be for $b$ or for $a$. 
In subsequent proofs we assume without loss of generality that it is a $b$-link and that (therefore) $m \geq 2$ is even.

\begin{lemma} \label{lemma.zero}
  The edge state $m'$ of model $M^m$ can be distinguished by an epistemic formula. 
\end{lemma}

\begin{proof}
We show that the state $m'$ in $M^m$ is the unique point satisfying the formula 
  $K_a L_{ba}^m K_b p \et \neg L_{ba}^{m-1} K_b p$ (see Lemma~\ref{lem:distinguishingFormula} for notation).

We can see this as follows. The denotation of $K_b p$ is $\{0,0'\}$. 
  Therefore, the denotation of $L_a K_b p$ is $\{0,0',1,1'\}$, 
  and the denotation of $L_b L_a K_b p$ is $\{0,0',1,1',2,2'\}$, 
  and in general the denotation of any $L_{ba}^k K_b p$ for $k \leq m$ is $\{0,0', \dots, k-1, (k-1)', k,k'\}$. 
  In particular, the denotation of $L_{ba}^m K_b p$ is $\{0,0', \dots, m-1, (m-1)', m,m'\}$.

Then, as $m'$ is an edge state, the denotation of $K_aL_{ba}^m K_b p$ is $\{0,0', \dots, m-1, (m-1)',m'\}$ (so, without state $m$), 
  as in state $m$ agent $a$ considers a state $m+1$ possible (where $L_{ba}^m K_b p$ is false), 
  but in the edge $m'$ agent $a$ does not consider another state possible. 

  Next, the denotation of $\neg L_{ba}^{m-1} K_b p$ is the complement of $\{0,0', \dots, m-1, (m-1)'\}$, i.e., $\{m,m',m+1,m+2,\dots\}$. 
  The denotation of the conjunction \[ \delta_{m'} \ \ := \ \ K_aL_{ba}^m K_b p \et \neg L_{ba}^{m-1} K_b p \] of these two formulas 
  is the intersection of these two sets: $\{0,0', \dots, m-1, (m-1)',m'\} \inter \{m,m',m+1,m+2,\dots\} = \{m'\}$, as required. 

This shows that $m'$ has a distinguishing formula $\delta_{m'}$. 
\end{proof}

\begin{lemma} \label{lemma.one}
$M^m_0 \not \models \Box(K_b K_a p \vel K_b \neg K_a p)$
\end{lemma}
\begin{proof}
As edge state $m'$ of model $M^m$ has a distinguishing formula, it follows from Lemma~\ref{lemma.dist} that any finite subset $T \subseteq S$ of model $M^m$ has a distinguishing formula. 
  In particular, we can therefore distinguish the set $T = \{0,0',1'\}$. 
  Let formula $\delta_T\in\langMl$ be such that $M^m_0 \models \delta_T$ and $\interpretation[M^m]{\delta_T} = \{0,0',1'\}$. 
  Note that $(M^m|\delta_T)_0 \not\models K_b K_a p \vel K_b \neg K_a p$. Therefore,  $M^m_0 \not \models \Box(K_b K_a p \vel K_b \neg K_a p)$.
\end{proof}

As an example of the Lemmas \ref{lemma.zero} and \ref{lemma.one}, consider model $M^2$. 
The distinguishing formula of world $2'$ is $\delta_{2'} = K_aL_bL_a K_b p \et \neg L_a K_b p$, 
and the submodel consisting of domain $\{0,0',1'\}$, using the method of Lemma \ref{lemma.dist}, has distinguishing formula 
$(L_b \delta_{2'} \et \neg \delta_{2'}) 
\vel (L_aL_b \delta_{2'} \et \neg L_b\delta_{2'}) 
\vel (L_bL_aL_b \delta_{2'} \et \neg L_aL_b\delta_{2'})$, 
which is equivalent to $L_bL_aL_b \delta_{2'} \et \neg \delta_{2'}$, i.e., to 
\[ L_bL_aL_b(K_aL_bL_a K_b p \et \neg L_a K_b p) \et \neg (K_aL_bL_a K_b p \et \neg L_a K_b p). \]
Similarly, we thus obtain that $\delta_T$ in the proof of Lemma \ref{lemma.one} is equivalent to the formula 
\[ L_{ab}^{m+2}(K_aL_{ba}^{m+1} K_b p \et \neg L_{ba}^m K_b p) \et \neg L_{ab}^{m-1}(K_aL_{ba}^{m+1} K_b p \et \neg L_{ba}^m K_b p).\]

\begin{lemma} \label{lemma.two}
$M^\omega_0 \models \Box(K_b K_a p \vel K_b \neg K_a p)$
\end{lemma}
\begin{proof}
Consider the relation $\bisimulation$ on $M^\omega$ defined as the symmetric and reflexive closure of $\{ (i,i') \mid i \in \Naturals \}$. 
  It is obvious that this relation $\bisimulation$ is a bisimulation (and even an isomorphism). 
  Differently said, the $0,1,\dots$ chain is the mirror image of the $0',1',\dots$ chain. Therefore, $M^\omega_0 \models \Box(K_b K_a p \vel K_b \neg K_a p)$: 
  firstly, any announcement must preserve actual state $0$ and therefore also preserves the bisimilar $0'$; 
  secondly, either $1$ and the bisimilar $1'$ are both eliminated by an announcement, after which $K_b K_a p$ is true at $0$, 
  or $1$ and $1'$ are both preserved, after which $K_b \neg K_a p$ is true at $0$.
\end{proof}

We continue by preparing the ground for the result that $M^\omega_0$ and $M^m_0$ cannot be distinguished in $\langPapal$ by a formula of epistemic depth at most $m$.
The main observation required for this result, is that given the very sparse structure of the model $M^\omega$, there are only three meaningful positive announcements. We can show this by looking in detail at the maximal refinements on $M^\omega$ and some of its restrictions.

\begin{lemma}\label{extrafortim}
Consider the relation $\bisimulation$ on $M^\omega$ consisting of: $(0,0)$, $(0,0')$, $(0',0)$, $(0',0')$, and all pairs $(i,j)$, $(i,j')$, $(i',j)$, and $(i',j')$ such that $i,j \in \Naturals$, $i,j >0$, and $i \leq j$. Then $\bisimulation$ is a refinement.
\end{lemma}

\begin{proof}
The {\bf atoms}-$p$ requirement is satisfied as all pairs in the relation only relate states wherein $p$ is true in both ($(0,0)$, $(0,0')$, $(0',0)$, and $(0',0')$) or wherein $p$ is false in both (all other pairs). 

For {\bf back}-$a$, let $(i,j') \in \bisimulation$ for $i,j \in \Naturals$ with $i,j >0$ and suppose $j' \sim_a k'$. We need to consider several cases: \begin{itemize}
\item if $k=j-1$ and $j= 1$, then: $i=1$ and $k' = 0'$ and as $0 \sim_a 1$ we choose $(0,0') \in \bisimulation$; 
\item if $k = j-1$ and $j\neq 1$, then: if $i \leq j-1$ then $(i,(j-1)') \in \bisimulation$ else $i=j$ and $(i-1) \sim_a i$ so $(i-1,(j-1)') \in \bisimulation$; 
\item if $k = j$, then (choose $i \sim_a i$ and) $(i,j') \in \bisimulation$; \item if $k = j+1$, then $i \leq j$ implies $i \leq j+1$ so $(i,(j+1)') \in \bisimulation$. 
\end{itemize}
The cases where $(i,j), (i',j'), (i',j) \in \bisimulation$ are similar, and the clause {\bf back}-$b$ can also be similarly proved. We further note that $\bisimulation$ is the maximal refinement on $M^\omega$.
\end{proof}

The proof of Lemma~\ref{extrafortim} also holds for certain connected submodels of $M^\omega$ such that $0$ or $0'$ and $1$ or $1'$ are in the connected part $S_x^y$. We recall the order $\preceq$ defined on $\Naturals \union \Naturals' \union \{\omega,\omega'\}$. These submodels are cases in the following Corollary~\ref{extrafortim2} and they are needed in the subsequent Lemma \ref{onlythreeannouncements}. (The submodels only containing $0$ or $0'$, or only excluding $0$ and $0'$, are less of interest, as will become clear in the proof of Lemma \ref{onlythreeannouncements}.)

\begin{corollary}\label{extrafortim2}
Let $\bisimulation'$ be the restriction of $\bisimulation$ to $S^y_x$ for some $x,y \in S \union \{\omega,\omega\}$, where $S = \domain(M^\omega)$. 
\begin{enumerate}
\item if $x \preceq 1'$ and $1 \preceq y$ then $\bisimulation'$ is a refinement on $M^\omega|S^y_x$. \label{extra1}
\item if $x =0$ then $\bisimulation'$ is a refinement on $M^\omega|S^y_x$.
\item if $x = 0'$ and $1 \preceq y$ then $\bisimulation' \setminus \{(0',0)\}$ is a refinement on $M^\omega|S^y_x$. \label{extra2}
\item if $x \preceq 1'$ and $y=0$ then $\bisimulation' \setminus \{(0,0')\}$ is a refinement on $M^\omega|S^y_x$. \label{extra3}
\item if $y =0'$ then $\bisimulation'$ is a refinement on $M^\omega|S^y_x$.
\end{enumerate}
\end{corollary}
In the third item above the pair $(0',0)$ is now excluded because {\bf back}-$a$ fails, as the link $0 \sim_a 1$ cannot be matched in state $0'$, wherein only $0' \sim_a 0'$. In the fourth item above the pair $(0,0')$ is now excluded because the link $0' \sim_a 1'$ cannot be matched in state $0$. 

\begin{lemma}\label{onlythreeannouncements}
  Let $\phi \in \langMlPlus$, $N$ be any submodel of $M^\omega$ and $t \in \domain(N)$. Then we have either $(N|\phi)_t\simeq (N|p)_t$, or $(N|\phi)_t\simeq (N|\lnot p)_t$, or $(N|\phi)_t\simeq N_t$.
\end{lemma}

\begin{proof}
  For the purposes of bisimulation it is sufficient to consider the connected component of $N|\phi$ containing $t$. 
  As the model $N$ is an $a$-$b$-chain this connected component will be a model $(M^\omega|S_x^y)_t$ for some $x,y\in S \union \{\omega,\omega'\}$ with $x\preceq t \preceq y$. So we have that $(N|\phi)_t \simeq (M^\omega|S_x^y)_t$. 
  Then we have the following cases:
     \begin{enumerate}
     \item If $x=0'$ and $y=0$ (or $x=y=0$, or $x=y=0'$), then $p$ is true everywhere, and the connected component is bisimilar to a singleton model wherein $p$ is true. An announcement of $p$ suffices here, so  $(N|\phi)_t\simeq (N|p)_t$.
     \item If $x,y\in\Naturals\setminus\{0\}$ or $x,y\in\Naturals'\setminus\{0'\}$, then $p$ is false everywhere. An announcement of $\lnot p$ will equally result in a model restriction only containing $\neg p$ states. Both restrictions are bisimilar to a singleton model wherein $p$ is false, so $(N|\phi)_t\simeq (N|\lnot p)_t$.
     \item Finally, as we are in a connected component, if neither of the above cases are true, 
       then we must have in our connected model a state $i\in\{0,0'\}$ and a state $j\in\{1,1'\}$ that were preserved by the announcement of $\phi$. 
       Now for every $k\in\domain(N)$ where $k\in(\Naturals\setminus\{0\})\cup(\Naturals'\setminus\{0'\})$ we have $N_k$ is a refinement of $N_j$ (Corollary \ref{extrafortim2}), 
         so every such $k$ must have been preserved by the announcement of $\phi$ (Lemma~\ref{refinement-preserves}). 
         Further: \begin{itemize} \item if $j = 1$ and $i=0'$ so that $0$ is also in $(N|\phi)$, then $N_{0'}$ is a refinement of $N_0$ (Corollary \ref{extrafortim2}.\ref{extra2}) so that both $0$ and $0'$ are preserved; 
         \item if $i = 0$ and $j=1'$ so that $0'$ is also in $(N|\phi)$, then $N_0$ is a refinement of $N_{0'}$ (Corollary \ref{extrafortim2}.\ref{extra3}) so that again both $0$ and $0'$ are preserved;
\item if $i = 0'$ but $0$ is not in $(N|\phi)$, then $0'$ was preserved by assumption;
\item if $i = 0$ but $0'$ is not in $(N|\phi)$, then $0$ was preserved by assumption.
\end{itemize}
       Therefore every state in the connected component containing $t$ is preserved by $\phi$ and $(N|\phi)_t\simeq N_t$.
     \end{enumerate}
  \end{proof}

\newcommand{\qisim}{\ensuremath{\mathfrak{Q}}}

\begin{lemma} \label{lemma.three}
  Let $M,N$ be submodels of $M^\omega$, $s \in \domain(M)$, $t \in \domain(N)$, and $k \in \Naturals$. If  $M_s \simeq^k N_t$, then  $M_s \equiv^k_\mathit{apal+} N_t$.
\end{lemma}
\begin{proof}
By induction on $\phi$ we show the equivalent proposition:
\begin{quote}
  Let $\phi\in\langPapal$, $M,N$ be submodels of $M^\omega$, $s \in \domain(M)$, $t \in \domain(N)$, and $d(\phi)\leq k$ where $k \in \Naturals$. 
  If $M_s \simeq^k N_t$, then  $M_s \models \phi$ iff $N_t \models \phi$.
\end{quote}
  We only show the relevant cases $K_a\phi$, $[\psi]\phi$, and $\allppas\phi$. As $k$-bisimilarity is a symmetric relation it suffices to show just one direction for each case.
  Let $\bisimulation^0 \supseteq \dots \supseteq \bisimulation^k$ be such that $\bisimulation^0: M_s \simeq^0 N_t$, \dots, $\bisimulation^k: M_s \simeq^k N_t$. 

  {\bf Case $K_a \phi$:} Suppose $d(K_a\phi) \leq k$. 
  We have $M_s \models K_a \phi$ if and only if for all $s'\sim_a s$, $M_{s'} \models \phi$.
  As $\bisimulation^k: M_s\simeq^k N_t$, for all $t'\sim_a t$ there is some $s'\sim_a s$ such that $\bisimulation^{k-1}: M_{s'}\simeq^{k-1}N_{t'}$.
  By the induction hypothesis we have for all $\psi$ where $d(\psi)\leq k-1$, $M_{s'}\models\psi$ implies $N_{t'}\models\psi$.
  As $d(\phi)\leq k-1$, it follows that for all $t'\sim_a t$, $N_{t'}\models\phi$.
  Therefore $N_t\models K_a\phi$. 


  {\bf Case $[\psi]\phi$:} Suppose $d([\psi]\phi)\leq k$, and $M_s\models[\psi]\phi$. By the definition of $d$ we may suppose that $d(\psi) = i$ and $d(\phi) = j$ where $i+j\leq k$. As $\bisimulation^k: M_s\simeq^k N_t$, by the induction hypothesis, $M_s\models\psi$ if and only if $N_t\models\psi$.
  Therefore, if $M_s\not\models\psi$ then $N_t\not\models\psi$ and vacuously $N_t\models[\psi]\phi$, as required.
  Suppose now that $M_s\models\psi$, so that also $N_t\models\psi$. 

  We define the following series of relations from $(M|\psi)$ to $(N|\psi)$ for $\ell=0,\hdots,k-i$: 
  $\qisim^\ell = \{(s,t)\in\bisimulation^{\ell+i}\mid M_s\models\psi\}$. Note that as $(s,t)\in\qisim^\ell$ implies $(s,t)\in\bisimulation^{\ell+i}$ for any such pair $(s,t)$, and $d(\psi)=i\leq \ell+i$, it follows by induction that $N_t\models\psi$, so indeed these are relations from $(M|\psi)$ to $(N|\psi)$.
  
We now show that the clauses {\bf atoms}, {\bf forth} and {\bf back} of bounded bisimulation (Definition~\ref{n-bisimulation}) hold for $\qisim^\ell = \qisim^0,\hdots,\qisim^{k-i}$, for any pair $(s,t) \in \qisim^\ell$.

Case $\ell =0$. We show {\bf atoms}-$p$, for $p \in P$. From $(s,t) \in \qisim^0$ it follows that $(s,t) \in \bisimulation^i$. As $\bisimulation^0 \supseteq \bisimulation^i$, it also follows that $(s,t) \in \bisimulation^0$, i.e., $s$ and $t$ satisfy the same atoms. 
Therefore $\qisim^0: (M|\psi)_s \simeq^0 (N|\psi)_t$.

Case $\ell >0$. We show $\ell$-{\bf forth}-$a$. Let $s \sim_a s'$ and $M_{s'} \models \psi$ (i.e., $s \sim_a s'$ in $(M|\psi)$). From $\bisimulation^{\ell+i}: M_s \simeq^{\ell+i} N_t$ and $s \sim_a s'$ follows that there is a $t' \sim_a t$ such that $\bisimulation^{\ell+i-1}: M_{s'} \simeq^{\ell+i-1} N_{t'}$. As $\ell>0$ and $d(\psi)=i$, $d(\psi)\leq \ell+i-1$. From $\bisimulation^{\ell+i-1}: M_{s'} \simeq^{\ell+i-1} N_{t'}$, $M_{s'} \models \psi$ and  $d(\psi)\leq \ell+i-1$ it follows by the induction hypothesis that $N_{t'} \models \psi$. Therefore $t'$ is in the domain of $N|\psi$. By definition, from $\bisimulation^{\ell+i-1}: M_{s'} \simeq^{\ell+i-1} N_{t'}$ it follows that $\qisim^{\ell-1}: (M|\psi)_{s'} \simeq^{\ell-1} (N|\psi)_{t'}$. Therefore, $t'$ satisfies the requirement for $\ell$-{\bf forth}-$a$ for relation $\qisim^\ell$. The clause $\ell$-{\bf back}-$a$ is shown similarly.

In particular, $(M|\psi)_s\simeq^{k-i}(N|\psi)_t$. From assumptions $M_s \models [\psi]\phi$ and $M_s \models \psi$ it follows that $(M|\psi)_s \models \phi$. Therefore, using that $d(\phi) = j \leq k-i$ and applying the induction hypothesis once again, we obtain that $(N|\psi)_t\models\phi$, which with $N_t \models \psi$ delivers the required $N_t\models[\psi]\phi$.


  {\bf Case $\allppas \phi$:} Suppose $d(\allppas\phi)\leq k$, and $M_s\models\allppas\phi$. Then $M_s\models[\psi]\phi$ for all $\psi\in\langMlPlus$.
  Now for any $\psi\in\langMlPlus$, by Corollary~\ref{onlythreeannouncements} we have either: $(N|\psi)_t\simeq(N|p)_t$, $(N|\psi)_t\simeq(N|\lnot p)_t$ or $(N|\psi)_t\simeq N_t$. 
  \begin{enumerate}
    \item In the first case, since $M_s\simeq^k N_t$, $N_t\models p$ implies $M_s\models p$ so both $(M|p)_s$ and $(N|p)_t$
      are bisimilar to the singleton model where $p$ is true. 
      As $M_s\models\allppas\phi$, we have $(M|p)_s\models\phi$ and thus $(N|p)_t\models\phi$ (Lemma~\ref{bisimulation-preserves-papal}).
      Since $(N|p)_t\simeq (N|\psi)_t$ we have $N_u\models[\psi]\phi$.
    \item The second case is similar: if $(N|\psi)_t\simeq (N|\lnot p)_t$, then $N_t\models\lnot p$ implies $M_s\models \lnot p$. 
      We then have that $(M|\lnot p)_s$ and $(N|\lnot p)_t$ are bisimilar to the singleton model where $p$ is false, and thus $(M|p)_s\models\phi$ implies $(N|p)_t\models\phi$.
      It follows that $N_t\models[\psi]\phi$.
    \item Finally, if $(N|\psi)_t\simeq N_t$ then since $M_s\simeq^k N_t$ and $M_s\models[\top]\phi$, 
      we have $M_s\models\phi$ and $N_t\models\phi$ by the induction hypothesis.
      Therefore $(N|\psi)_t\models\phi$ and $N_t\models[\psi]\phi$.
  \end{enumerate}
  Therefore, for every $\psi\in \langMlPlus$, $N_t\models[\psi]\phi$, so $N_t\models\allppas\phi$ as required.
  \end{proof}

As $M^m_0 \simeq^m M^\omega_0$, the following corollary is rather a special case of the previous lemma.
\begin{corollary} \label{final}
Let $\phi \in \langPapal$ such that $d(\phi) \leq m$. Then $M^m_0 \models \phi$ iff  $M^\omega_0 \models \phi$.
\end{corollary}

\begin{theorem}\label{expressivity-s5}
    \logicPapal{} is not at least as expressive as \logicApal{} for multiple agents.
\end{theorem}

\begin{proof}
  Recall the $\langApal$ formula $\Box(K_b K_a p \vel K_b \neg K_a p)$ from Lemmas~\ref{lemma.one}~and~\ref{lemma.two}. 
  Let us assume that there is an equivalent formula $\phi \in \langPapal$ and the epistemic depth of this formula is $d(\phi) = m$. 
  We recall that the epistemic depth counts the number of stacked knowledge modalities, but ignores the arbitrary (positive) announcement modalities.
  
  Consider again the models $M^m_0$ and $M^\omega_0$ of Figure \ref{expressivity-s5i-1mm}. We have shown the following:
  \begin{enumerate}
    \item $M^m_0 \not \models \Box(K_b K_a p \vel K_b \neg K_a p)$ (Lemma \ref{lemma.one})
    \item $M^\omega_0 \models \Box(K_b K_a p \vel K_b \neg K_a p)$ (Lemma \ref{lemma.two})
    \item $M^m_0 \models \phi$ iff  $M^\omega_0 \models \phi$ (Corollary \ref{final})
  \end{enumerate} 
  The assumption that $\Box(K_b K_a p \vel K_b \neg K_a p)$ is equivalent to $\phi$ is in contradiction with these results. Therefore, no such equivalent $\phi$ exists.
\end{proof}

\subsection{APAL is not at least as expressive as APAL$^+$}

We wish to establish incomparability of \logicApal\ and \logicPapal, so it remains to show that \logicApal{} is not at least as expressive as \logicPapal{} for multiple agents. This we will do in the following Theorem \ref{expressivity-s52}, by, again, the standard method of providing two pointed epistemic models and a formula (in \langPapal) such that the models can be distinguished by that formula but cannot be distinguished by any formula in the other language (in \langApal). Before that theorem we will introduce the models used in its proof, present an intuitive example to illustrate the proof method, and introduce some lemmas to be used in that proof.

Consider models $M^l_0$ and $N^r_0$ in Figure \ref{fig.grappig}. Both are $a$-$b$-chains, and such that a variable $p$ is false in the evaluation point $0$ and the values of $p$ are swapped in adjoining states. However, the model $M^l$ terminates on the $a$-link side of the designated state $0$ in state $l$ (for \emph{left}) and is infinite on $b$-link side of $0$, whereas the model $N^r$ terminates on the $b$-link side of $0$ in state $r$ (for \emph{right}) and is infinite on the left.

Formally, let $l$ be a negative odd integer and let $r$ be a positive even integer, then the domain of $M^l$ is $\{ i \mid i \in \integers, i \geq l\}$. Relation $R_a$ in $M^l$ is the symmetric and reflexive closure of $\{ ({2i-1},{2i}) \mid i \in \integers, 2i-1 \geq l\}$, whereas $R_b$ is the symmetric and reflexive closure of $\{ ({2i}, {2i+1}) \mid i \in \integers, 2i > l\}$, and $V(p) = \{ {2i+1} \mid i \in \integers, 2i+1 \geq l \}$. Then, the domain of $N^r$ is $\{ i \mid i \in \integers, i \leq r\}$; the relations and valuation in $N^r$ are similarly defined as in $M^l$. We recall that $M^l$ and $N^r$ are $a$-$b$-chains. Both have a single edge.

\begin{figure}[ht]
 \begin{tikzpicture}[>=stealth',shorten >=1pt,auto,node distance=5em,thick]

          \node (s) {$\{p\}$};
          \node (sr) [right of=s] {$\{\}$};
          \node (srr) [right of=sr] {\color{white}$\{\}$};
          \node (sl) [label=above:{$0$}] [left of=s] {\underline{$\{\}$}};
           \node (sll) [left of=sl] {$\{\}$};
          \node (slll)  [label=above:{$l$}] [left of=sll] {$\{p\}$};
         \draw (slll) edge node {$\agentA$} (sll);
         \draw[dashed]  (sll) edge (sl);  
           \draw  (sl) edge node {$\agentB$} (s);
           \draw  (s) edge node {$a$}  (sr);
           \draw[dashed]  (sr) edge (srr);

          \node (bs) [below of=s,node distance=3em]{$\{\}$};
          \node (bsr)  [right of=bs] {$\{p\}$};
          \node (bsrr) [label=below:{$r$}] [right of=bsr] {$\{\}$};
          \node (bsl) [label=below:{$0$}] [left of=bs] {\underline{$\{\}$}};
          \node (bsll) [left of=bsl] {$\{p\}$};
          \node (bslll) [left of=bsll] {\color{white} $\{p\}$};
         \draw[dashed] (bslll) edge (bsll);
         \draw  (bsll) edge  node {$a$} (bsl);
           \draw[dashed]  (bsl) edge (bs);  
           \draw  (bs) edge node {$\agentB$} (bsr);
           \draw  (bsr) edge node {$\agentA$} (bsrr);
 
\node (ll) [left of=slll] {\color{white} $\{p\}$};
\node (bll) [left of=bslll] {\color{white} $\{p\}$};
\node (brr) [right of=bsrr] {\color{white} $\{p\}$};
\node (m) [left of=slll] {$M_0^l$:};
\node (n) [left of=bslll] {$N_0^r$:};

       \end{tikzpicture}
\caption{Models used in the proof of Theorem \ref{expressivity-s52}}
\label{fig.grappig}
\end{figure}
 
\begin{figure}[ht]
 \begin{tikzpicture}[>=stealth',shorten >=1pt,auto,node distance=5em,thick]

          \node (s) {$\{p\}$};
          \node (sr) [right of=s] {$\{\}$};
          \node (srr) [right of=sr] {\color{white} $\{p\}$};
          \node (sl) [label=above:{$0$}] [left of=s] {\underline{$\{\}$}};
          \node (sll) [label=above:{$-1$}] [left of=sl] {$\{p\}$};
          \node (slll) [left of=sll] {\color{white} $\{\}$};
         \draw  (sll) edge node {$\agentA$} (sl);
           \draw  (sl) edge node {$\agentB$} (s);
           \draw  (s) edge node {$\agentA$} (sr);
           \draw[dashed]  (sr) edge (srr);

          \node (bs) [below of=s,node distance=3em]{$\{p\}$};
          \node (bsr) [label=below:{$2$}] [right of=bs] {$\{\}$};
          \node (bsrr) [right of=bsr] {\color{white} $\{p\}$};
          \node (bsl) [label=below:{$0$}] [left of=bs] {\underline{$\{\}$}};
          \node (bsll) [left of=bsl] {$\{p\}$};
          \node (bslll) [left of=bsll] {\color{white} $\{\}$};
         \draw[dashed] (bslll) edge (bsll);
         \draw  (bsll) edge node {$\agentA$} (bsl);
           \draw  (bsl) edge node {$\agentB$} (bs);
           \draw  (bs) edge node {$\agentA$} (bsr);
\node (corr) [left of=slll] {\color{white} $\{p\}$};
\node (bcorr) [left of=bslll] {\color{white} $\{p\}$};
\node (m) [left of=slll] {$M^{-1}_0$:};
\node (n) [left of=bslll] {$N^2_0$:};
       \end{tikzpicture}
\caption{An example for $l=-1$ and $r=2$}
\label{fig.grappig2}
\end{figure}

In order to informally explain the method in the subsequent proof, first consider models $M^{-1}_0$ and $N^{2}_0$ in Figure \ref{fig.grappig2}. In $M^{-1}$ but not in $N^2$, from the evaluation point $0$, the $a$-link is closer to the edge than the $b$-link. The formula $\allppas (L_b p \imp L_a p)$ formalizes this property in \langPapal. 

In $M^{-1}$, the prefixes of this chain are defined by the (positive) formulas: $K_b p$ (for $\{-1\}$), $K_a (\neg p \vel K_b p)$ (for $\{-1,0\}$), $K_b (p \vel K_a (\neg p \vel K_b p))$ (for $\{-1,0,1\}$), etc. As we build these prefixes from the left, the $a$-link from $0$ is included before the $b$-link from $0$ is included. There are yet other positively definable subsets containing $0$, such as the $\neg p$-states. But that cuts off both links. Differently said, if the $b$-link from state $0$ to state $1$ is included then the $a$-link from state $0$ to state $-1$ is included. And both have a different value of $p$ than in $0$. This gives us $L_b p \imp L_a p$. And therefore, $M^{-1}_0 \models \allppas (L_b p \imp L_a p)$. 

Now look at $N^2$. There, similar reasoning makes us conclude that the $b$-link is always included before the $a$-link. So we can make a positive announcement, namely $K_b (\neg p \vel K_a (p \vel K_b \neg p))$, resulting in the restriction to $\{0,1,2\}$, after which $L_b p$ is true but $L_a p$ is false. So $N^2_0 \not \models \allppas (L_b p \imp L_a p)$. 

Of course the models $M^{-1}_0$ and $N^2_0$ can be easily distinguished in \langApal{} too. They can even be distinguished in \langMl, without \logicApal{} quantifiers, for example by a formula expressing that the distance to the edge is $1$ in $M^{-1}_0$ but more than $1$ in $N^2_t$. As $K_b p$ distinguishes state $-1$ in $M^{-1}$, this formula is $L_a K_b p$. We note that $M^{-1}_0 \models L_a K_b p$ whereas $N^2_0 \not\models L_a K_b p$. But, tellingly, you need to have that distance explicitly in the formula, unlike in the \langPapal{} formula. And $d(L_a K_b p) = 2$, larger than $d(\allppas (L_b p \imp L_a p)) = 1$. 

Having prepared the ground for the proof, we now present Theorem \ref{expressivity-s52} (at the end of this section) and preceding lemmas.

\begin{lemma} \label{lemma2.zero}
The positively definable restrictions of $M^l$ are: all states, the $p$-states, the $\neg p$-states, any finite prefix of the $a$-$b$-chain $M^l$, and the union or intersection of any of the previous.
\end{lemma}
\begin{proof}
The relation $\bisimulation := \{({i+2j},i) \mid j \in \Naturals, i \in \integers, i \geq l \}$ is the maximal refinement on $M^l$. It is a refinement because $M^l_{{i+2j}} \succeq M^l_{i}$ iff $M^l_{i}$ is isomorphic to a submodel of $M^l_{{i+2j}}$. A submodel is the most typical example of the structural loss represented by a refinement. The relation $\bisimulation$ is also maximal. We cannot pair a $p$-state to a larger $p$-state, such as in $(l, l+2)$: {\bf back}-$b$ would then fail: from $l+2$ we can reach a $\neg p$-state via $l+1 \sim_b l+2$, but we cannot reach a $\neg p$-state by a $b$-link from state $l$. Similarly we cannot have any other pair where the second argument is a state named with a larger number than the first argument, by iterating {\bf back} steps. 

Given $\bisimulation$, the subsets of the domain of $M^l$ that are closed under refinement are: all states, the $p$-states, the $\neg p$-states, and the finite prefixes of the chain $M^l$. To this we further add the union or intersection of any of the previous, where we note that the union of two prefixes is the longer prefix and the intersection of two prefixes is the smaller prefix. This means that also closed under refinement are: the $p$-states of any finite prefix, the $\neg p$-states of any finite prefix, and the union of a prefix of the chain with the set of $p$-states, or $\neg p$-states, of a larger prefix (such as the set $\{-1,0,1,2,3,5,7,9\}$).

We now show that all refinement closed subsets of the domain of $M^l$ are positively definable. This is not evident, as the domain of $M^l$ is not finite (so Lemma~\ref{refinement-closed-positive-announcements} does not apply). We define: $\delta^l_l := K_b p$, $\delta^l_{i+1} := K_a (\neg p \vel \delta^l_i)$ for $i$ an odd natural number, and $\delta^l_{i+1} := K_b (p \vel \delta^l_i)$ for $i$ an even natural number. The other positive formulas defining refinement closed subsets are conjunctions or disjunctions of the previous; none of those however will play a role in the continuation.
\end{proof}

The argument is the same for the model $N^r$. In this case relation $\bisimulation' := \{({i-2j},i) \mid j \in \Naturals, i \in \integers, i \leq r \}$ is the maximal refinement on $N^r$, and any $N^r_{i}$ is isomorphic to a submodel of $N^r_{{i-2j}}$. The positive formulas defining the prefixes are now defined as: $\delta^r_r := K_b p$, $\delta^r_{i-1} := K_a (\neg p \vel \delta^r_i)$ for $i$ an even natural number, and $\delta^r_{i-1} := K_b (p \vel \delta^r_i)$ for $i$ an odd natural number. 

\begin{corollary} \label{lemma2.zero2}
The  positively definable restrictions of $N^r$ are: all states, the $p$-states, the $\neg p$-states, any finite prefix of the $a$-$b$-chain $N^r$, and the union or intersection of any of the previous.
\end{corollary}

\begin{lemma} \label{lemma2.one}
$M^l_0 \models \allppas (L_b p \imp L_a p)$
\end{lemma}
\begin{proof}
Let $T \subseteq \domain(M)$ be positively definable and such that $0 \in T$. Then either $M^l|T$ is a prefix of $M^l$ containing $0$, or $M^l|T$ consists of disconnected parts of which $M^l|\{0\}$ is a singleton part. In the second case, from $(M^l|\{0\})_0 \models \neg L_a p$ and $(M^l|\{0\})_0 \models \neg L_b p$ follows $(M^l|\{0\})_0 \models L_a p \imp L_b p$. In the first case, as $M^l|T$ is a prefix of $M$ containing $0$, the $a$-link to $-1$ (where $-1$ may be $l$) must always be included in that restriction if the $b$-link to $1$ is included. Therefore $(M^l|T)_0 \models L_b p \imp L_a p$. From $(M^l|T)_0 \models L_b p \imp L_a p$ for all $T$ containing $0$, and the observation that all such $T$ are positively definable (Lemma~\ref{lemma2.zero}), it follows that $M^l_0 \models \allppas (L_b p \imp L_a p)$.
\end{proof}

\begin{lemma} \label{lemma2.two}
$N^r_0 \not\models\allppas (L_b p \imp L_a p)$
\end{lemma}
\begin{proof}
The prefix $T = \{0, \dots, r \}$ of $N^r$ is positively definable by $\delta^r_0\in\langMlPlus$ (see above). We now have that $(N^r|\delta^r_0)_0 \models L_b p$, because $0 \sim_b 1$ and $(N^r|\delta^r_0)_{1} \models p$, but $(N^r|\delta^r_0)_0 \not\models L_a p$, because state ${-1}$ (and any other state $i < -1$) has been eliminated by the announcement of $\delta^r_0$. Therefore, $(N^r|\delta^r_0)_0 \models L_b p \et\neg L_a p$. From that and $N^r_0 \models \delta^r_0$ (as $0 \in T$) it follows that $N^r_0 \models \dia{\delta^r_0} (L_b p \et \neg L_a p)$. Therefore $N^r_0 \models \someppas (L_b p \et \neg L_a p)$, i.e., $N^r_0 \not\models\allppas (L_b p \imp L_a p)$.
\end{proof}

The following lemma is very crucial. Note that the restrictions below can be for any subset of the domain, not necessarily positively definable.

\begin{lemma} \label{lemma2.twofifty}
Given are restricted models $M$ of $M^l$ and $N$ of $N^r$, and $i,j \in \Naturals$ with $i \in \domain(M)$ and $j \in \domain(N)$. If $M_i \simeq^n N_j$, then for all $\psi\in\langMl$ such that $M_i \models \psi$ there is a $\psi'\in\langMl$ such that $(M|\psi)_i \simeq^n (N|\psi')_j$, and for all $\psi'\in\langMl$ such that $N_j \models \psi'$ there is a $\psi\in\langMl$ such that $(M|\psi)_i \simeq^n (N|\psi')_j$.
\end{lemma}
\begin{proof}

Given $\psi\in\langMl$ with $M_i \models \psi$, let $M'_i$ be obtained by restricting $(M|\psi)_i$ to states at most $n$ steps, on either side, from $i$, and omitting components disconnected from $i$. We then have that $M'_i \simeq^n (M|\psi)_i$, and that $M'_i$ is a finite chain of length at most $2n+1$. We recall that any finite subset in $N^r$ is distinguishable in $\langMl$, using the distance from endpoint $r$ (see Lemma~\ref{lemma.dist}). Similarly, any finite subset in a connected part of $N$ is distinguishable in $\langMl$ from its complement in that part (which again follows from  Lemma~\ref{lemma.dist} or otherwise from Lemma~\ref{lem:distinguishingFormula}). So, as $M' \subseteq M$ and $M_i \simeq^n N_j$, there is a $\psi'\in\langMl$ and a finite $N' \subseteq N$ such that $N'_j \simeq (N|\psi')_j$ (i.e., unbounded) and $M'_i \simeq^n N'_j$. From that and $M'_i \simeq^n (M|\psi)_i$ it follows that $(M|\psi)_i \simeq^n (N|\psi')_j$. The proof in the other direction, assuming a $\psi'\in\langMl$ such that $N_j \models \psi'$, is similar.
\end{proof}
It is important to note that in the above proof the epistemic depths $d(\psi)$ and $d(\psi')$ are not related to $n$: they are arbitrary and therefore can be larger than $n$.

\begin{lemma} \label{lemma2.three}
Let $M \subseteq M^l$, $N \subseteq N^r$, $i,j \in \Naturals$ with $i \in \domain(M)$ and $j \in \domain(N)$, and $n \in \Naturals$: if $M_i \simeq^n N_j$, then $M_i \equiv^n_{\mathit apal} N_j$.
\end{lemma}
\begin{proof}
  We show the equivalent formulation:
  \begin{quote}
  For all $\phi\in\langApal$, $M \subseteq M^l$, $N \subseteq N^r$, $i,j \in \Naturals$ with $i \in \domain(M)$ and $j \in \domain(N)$, and $n \in \Naturals$: 
    if $M_i \simeq^n N_j$ and $d(\phi) \leq n$, then $M_i \models \phi$ iff $N_j \models \phi$.
  \end{quote}
  The proof is by induction on the structure of $\phi$. 
  The cases of interest are $K_a \phi$, $[\psi]\phi$, and $\Box \phi$. 
  The first two cases are similar to those shown in Lemma \ref{lemma.three}, and therefore shown in less detail. 
  As $n$-bisimilarity is a symmetric relation, it suffices to show only one direction of the equivalence.

  \bigskip

  {\bf Case $K_a \phi$:} Suppose $d(K_a\phi) \leq n$. 
  We have $M_i \models K_a \phi$ if and only if for all $i'\sim_a i$, $M_{i'} \models \phi$.
  As $M_i\simeq^n N_j$, for all $j'\sim_a j$ there is some $i'\sim_a i$ such that $M_{i'}\simeq^{n-1}N_{j'}$.
  By the induction hypothesis, given $d(\phi)\leq n-1$, we have for all $j'\sim_a j$, $N_{j'}\models\phi$.
  Therefore $N_j\models K_a\phi$. 

  \bigskip

  {\bf Case $[\psi]\phi$:} Suppose $d([\psi]\phi)\leq n$, and $M_i\models[\psi]\phi$. 
  By the definition of $d$ we may suppose that $d(\psi) = x$ and $d(\phi) = y$ where $x+y\leq n$.
  Let $\bisimulation^0 \supseteq \dots \supseteq \bisimulation^n$ be such that $\bisimulation^0: M_i \simeq^0 N_j$, \dots, $\bisimulation^n: M_i \simeq^n N_j$. 
  For all $(i',j')\in \bisimulation^x$, we have $M_{i'}\simeq^xN_{j'}$, so by the induction hypothesis, $M_{i'}\models\psi$ if and only if $N_{j'}\models\psi$.
  Therefore, if $M_i\not\models\psi$ then $N_j\not\models\psi$ and vacuously $N_j\models[\psi]\phi$, as required.
  Suppose now that $M_i\models\psi$. 
  We define the series of relations from $(M|\psi)$ to $(N|\psi)$ for $z=0,\hdots,y$: 
  $\qisim^z = \{(i',j')\in\bisimulation^{n-z}\mid M_{i'}\models\psi\}$. 
  The conditions {\bf atoms}, {\bf forth} and {\bf back} for the bounded bisimulation of Definition~\ref{n-bisimulation} hold for $\qisim^0,\hdots,\qisim^y$,
  and so $(M|\psi)_i\simeq^y(N|\psi)_j$.
  Applying the induction hypothesis once again, we have $(M|\psi)_i\models\phi$ implies $(N|\psi)_j\models\phi$, and so $N_j\models[\psi]\phi$.

  \bigskip

{\bf Case $\allpas \phi$:}

To match the previous lemma, we show the dual diamond form.

\bigskip

\noindent $
M_i \models \Dia \phi \\ \Eq \\ 
\text{there is } \psi \in \langMl, M_i \models \dia{\psi}\phi \\ \Eq \\ 
\text{there is  } \psi \in \langMl, M_i \models \psi \text{ and } (M|\psi)_i \models \phi \\ \Eq \hfill \text{Lemma \ref{lemma2.twofifty}} \\ 
\text{there is  } \psi' \in \langMl, N_j \models \psi' \text{ and } (N|\psi')_j \models \phi \\ \Eq \\ 
\text{there is } \psi' \in \langMl, N_j \models \dia{\psi'}\phi \\ \Eq \\ 
N_j \models \Dia\psi$
\end{proof}

\begin{corollary} \label{lemma2.four}
Let $\phi \in \langApal$, $l < -d(\phi)$ and $r > d(\phi)$. Then $M^l_0 \models\phi$ iff $N^r_0 \models \phi$.
\end{corollary}

\begin{theorem}\label{expressivity-s52}
    \logicApal{} is not at least as expressive as \logicPapal{} for multiple agents.
\end{theorem}

\begin{proof}
Consider the formula $\allppas (L_b p \imp L_a p)$. Let us suppose that $\allppas (L_b p \imp L_a p)$ is equivalent to a $\langApal$ formula $\phi$. The epistemic depth of this formula is $d(\phi)$. Let $M^l_0$ and $N^r_0$ be such that $|l|,r > d(\phi)$. Then:
\begin{enumerate}
\item  $M^l_0 \models \allppas (L_b p \imp L_a p)$ (Lemma \ref{lemma2.one});
\item $N^r_0 \not\models\allppas (L_b p \imp L_a p)$ (Lemma \ref{lemma2.two});
\item $M^l_0 \models\phi$ iff $N^r_0 \models \phi$ (Corollary \ref{lemma2.four}).
\end{enumerate}
This is a contradiction. Therefore, no such equivalent $\langApal$ formula exists.
\end{proof}

\begin{corollary}
\logicApal{} and \logicPapal{} have incomparable expressivity.
\end{corollary}

\begin{proof}
From Theorem \ref{expressivity-s5} and Theorem \ref{expressivity-s52}.
\end{proof}

The relative expressivity of \logicPapal{} to group announcement logic and coalition announcement logic, mentioned in the introduction, has recently been addressed in \cite{frenchetal:2019,Galimullin19}. It is shown that {\it GAL} is not at least as expressive as {\it CAL} and that \logicApal{} is not at least as expressive as {\it CAL}, with chain models for three agents instead of the two agent $a$-$b$-chains in our contribution. Whether {\it CAL} is not at least as expressive as {\it GAL} is an open question.

\section{Axiomatisation}\label{axiomatisation}

In this section we provide a sound and complete axiomatisation for arbitrary
positive announcement logic. It is as the (infinitary) axiomatisation for arbitrary public announcement logic given by
Balbiani {\em et al.}~\cite{balbianietal:2008,balbianietal:2015}, but with
restrictions to positive announcements in appropriate axioms.

\begin{definition}
    Consider a new symbol $\sharp$. The {\em necessity forms} are defined inductively as:
    $$\psi(\sharp) ::= \sharp \mid (\phi \implies \psi(\sharp)) \mid \announceA{\phi} \psi(\sharp) \mid \knows[\agentA] \psi(\sharp)$$
    where $\phi \in \langPapal$ and $\agentA \in \agents$.
\end{definition}
A necessity form contains a unique occurrence of the symbol $\sharp$. If $\psi(\sharp)$ is a necessity form and $\phi \in \langPapal$, then $\psi(\phi) \in \langPapal$, where $\psi(\phi)$ stands for the substitution of the unique occurrence of $\sharp$ in $\psi(\sharp)$ by $\phi$. We also call $\psi(\phi)$ an {\em instantiation} of $\psi(\sharp)$.


The axiomatisation \axiomPapal{} is given below. A formula is a {\em theorem} if it belongs to the least set of formulas containing all axioms and closed under the derivation rules. 

\begin{definition}
    The axiomatisation \axiomPapal{} consists of the following axioms and rules. In the rule ${\bf R+^\omega}$, the expressions $\chi(\announceA{\psi} \phi)$ and $\chi(\allppas \phi)$ are instantiations of a necessity form $\chi(\sharp)$.
    $$
    \begin{array}{llll}
        {\bf P} & \text{All propositional tautologies} &
        {\bf K} &  \knows[\agentA] (\phi \implies \psi) \implies (\knows[\agentA] \phi \implies \knows[\agentA] \psi)\\
        {\bf T} &  \knows[\agentA] \phi \implies \phi&
        {\bf 4} &  \knows[\agentA] \phi \implies \knows[\agentA] \knows[\agentA] \phi\\
        {\bf 5} &  \neg \knows[\agentA] \phi \implies \knows[\agentA] \neg \knows[\agentA] \phi&
        {\bf AP} &  \announceA{\phi} \atomP \iff (\phi \implies \atomP)\\
        {\bf AN} &  \announceA{\phi} \neg \psi \iff (\phi \implies \neg \announceA{\phi} \psi)&
        {\bf AC} &  \announceA{\phi} (\psi \land \chi) \iff (\announceA{\phi} \psi \land \announceA{\phi} \chi)\\
        {\bf AK} &  \announceA{\phi} \knows[\agentA] \psi \iff (\phi \implies \knows[\agentA] \announceA{\phi} \psi)&
        {\bf AA} &  \announceA{\phi} \announceA{\psi} \chi \iff \announceA{\phi \land \announceA{\phi} \psi} \chi\\
        {\bf A+} &  \allppas \phi \implies \announceA{\psi} \phi \text{ where } \psi \in \langMlPlus&
        {\bf MP} & \text{From }  \phi \text{ and }  \phi \implies \psi \text{ infer }  \psi\\
        {\bf NecK} & \text{From }  \phi \text{ infer }  \knows[\agentA] \phi&
        {\bf NecA} & \text{From }  \phi \text{ infer }  \announceA{\psi} \phi\\
        &&{\bf R+^\omega} &\text{From }  \chi(\announceA{\psi} \phi) \text{ for every } \psi \in \langMlPlus \text{ infer } \chi(\allppas \phi)    \end{array}
    $$
\end{definition}

The axiomatisation \axiomPapal{} is identical to the axiomatisation $\axiomApal^\omega$ in \cite{balbianietal:2008} and to the axiomatisation \axiomApal\ in \cite{balbianietal:2015}, except for the replacement of the \logicApal\ $\Box$ by the \logicPapal\ $\allppas$ on two occasions, resulting in the axiom {\bf A+}
    and the rule ${\bf R+^\omega}$. Other, non-essential differences are the different names for axioms and rules, for example the axiom we call {\bf K} they call $A1$, the axiom we call {\bf T} they call $A4$, and so on; and the presence of additional, known to be derivable, axioms in \cite{balbianietal:2015}.

Note that the proof of completeness of \axiomApal\ given in \cite{balbianietal:2008} was wrong and that a correct proof of completeness has been given in \cite{balbianietal:2015,philippe.corrected:2015}.

\begin{theorem}
    The infinitary axiomatisation \axiomPapal{} is sound and complete for the logic \logicPapal{}.
\end{theorem}

\begin{proof}
    The soundness of the axiomatisation is evident as the axiom {\bf A+}
    and the rule ${\bf R+^\omega}$ follow the semantics of the $\allppas$
    operator (just as their non-positive counterparts followed the semantics of the $\Box$ operator), and all remaining axioms and rules are, as well-known, standard from
    epistemic logic and public announcement logic.

    The completeness proof proceeds exactly as in~\cite{balbianietal:2015}, with appropriate restrictions from epistemic announcements to
    positive announcements in the cases of {\bf A+} and ${\bf R+^\omega}$.

    More precisely, the
    positive arbitrary announcement operator $\allppas$ only features in
    the subinductive case $\announceA{\psi} \allppas \chi$ and in the
    inductive case $\allppas \psi$ of the proof of the Truth Lemma. The Truth Lemma for \logicApal\ is proved by a complexity measure wherein $[\psi]\phi$ is less complex than $\Box\phi$ for any $\psi\in\langMl$. Similarly, $[\psi]\phi$ is less complex than $\allppas\phi$ for any $\psi\in\langMlPlus$. This justifies that substituting `epistemic' for `positive' in appropriate places is sufficient.
    
No other changes are required.   
    \end{proof}

We note that \axiomPapal{} is an infinitary axiomatisation, as the rule 
${\bf R+^\omega}$ requires an infinite number of premises. Just as for the infinitary axiomatisation of the logic \logicApal, it is unknown if a finitary axiomatisation exists.

\section{Conclusion}\label{future-work}

We presented a
    variant of arbitrary public announcement logic called {\em positive arbitrary
    public announcement logic}, \logicPapal, which restricts arbitrary public
    announcements to announcement of {\em positive formulas}. We showed that the model checking complexity of \logicPapal{} is PSPACE-complete, that \logicPapal{} is more
    expressive than public announcement logic \logicPal, that it is
    incomparable with \logicApal, and we provided a sound and complete infinitary axiomatisation. 
The proof of the decidability of \logicPapal{} is reported in a companion paper \cite{papaldec:2018}.

\bibliographystyle{abbrv}
\bibliography{biblio2020}

\providecommand{\noopsort}[1]{}
\begin{thebibliography}{10}

\bibitem{agotnesetal.jal:2010}
T.~{\AA}gotnes, P.~Balbiani, H.~van Ditmarsch, and P.~Seban.
\newblock Group announcement logic.
\newblock {\em Journal of Applied Logic}, 8:62--81, 2010.

\bibitem{agotnesetal:2016}
T.~{\AA}gotnes, H.~van Ditmarsch, and T.~French.
\newblock The undecidability of quantified announcements.
\newblock {\em Studia Logica}, 104(4):597--640, 2016.

\bibitem{andrekaetal:1998}
H.~Andr\'eka, I.~N\'emeti, and J.~van Benthem.
\newblock Modal languages and bounded fragments of predicate logic.
\newblock {\em Journal of Philosophical Logic}, 27(3):217--274, 1998.

\bibitem{philippe.corrected:2015}
P.~Balbiani.
\newblock Putting right the wording and the proof of the {T}ruth {L}emma for
  {APAL}.
\newblock {\em Journal of Applied Non-Classical Logics}, 25(1):2--19, 2015.

\bibitem{balbianietal:2008}
P.~Balbiani, A.~Baltag, H.~van Ditmarsch, A.~Herzig, T.~Hoshi, and T.~D. Lima.
\newblock `{K}nowable' as `known after an announcement'.
\newblock {\em Review of Symbolic Logic}, 1(3):305--334, 2008.

\bibitem{balbianietal:2015}
P.~Balbiani and H.~van Ditmarsch.
\newblock A simple proof of the completeness of {APAL}.
\newblock {\em Studies in Logic}, 8(1):65--78, 2015.

\bibitem{blackburnetal:2001}
P.~Blackburn, M.~de~Rijke, and Y.~Venema.
\newblock {\em Modal Logic}.
\newblock Cambridge University Press, Cambridge, 2001.
\newblock Cambridge Tracts in Theoretical Computer Science 53.

\bibitem{bozzellietal.inf:2014}
L.~Bozzelli, H.~van Ditmarsch, T.~French, J.~Hales, and S.~Pinchinat.
\newblock Refinement modal logic.
\newblock {\em Information and Computation}, 239:303--339, 2014.

\bibitem{browneetal:1987}
M.~Browne, E.~Clarke, and O.~Gr\"umberg.
\newblock Characterizing {K}ripke structures in temporal logic.
\newblock In H.~Ehrig, R.~Kowalski, G.~Levi, and U.~Montanari, editors, {\em
  TAPSOFT '87}, LNCS 249, pages 256--270. Springer, 1987.

\bibitem{chandraetal:1981}
A.~Chandra, D.~Kozen, and L.~Stockmeyer.
\newblock Alternation.
\newblock {\em Journal of the ACM}, 28:114--33, 1981.

\bibitem{charrieretal:2015}
T.~Charrier and F.~Schwarzentruber.
\newblock Arbitrary public announcement logic with mental programs.
\newblock In {\em Proc.\ of {AAMAS}}, pages 1471--1479. {ACM}, 2015.

\bibitem{Chaum:1988}
D.~Chaum.
\newblock The dining cryptographers problem: Unconditional sender and recipient
  untraceability.
\newblock {\em J. Cryptol.}, 1(1):65--75, 1988.

\bibitem{cordonetal.tcs:2013}
A.~Cord{\'o}n-Franco, H.~van Ditmarsch, D.~Fern{\'a}ndez-Duque, and
  F.~Soler-Toscano.
\newblock A colouring protocol for the generalized russian cards problem.
\newblock {\em Theor. Comput. Sci.}, 495:81--95, 2013.

\bibitem{dagostinoetal:2000}
G.~d'Agostino and M.~Hollenberg.
\newblock Logical questions concerning the $\mu$-calculus: Interpolation,
  {L}yndon and {L}os-{T}arski.
\newblock {\em Journal of Symbolic Logic}, 65(1):310--332, 2000.

\bibitem{fischeretal:1996}
M.~Fischer and R.~Wright.
\newblock Bounds on secret key exchange using a random deal of cards.
\newblock {\em Journal of Cryptology}, 9(2):71--99, 1996.

\bibitem{French06}
T.~French.
\newblock Bisimulation quantified modal logics: Decidability.
\newblock In G.~Governatori, I.~Hodkinson, and Y.~Venema, editors, {\em Proc.\
  of 6th {AiML}}, pages 147--166. College Publications, 2006.

\bibitem{frenchetal:2019}
T.~French, R.~Galimullin, H.~van Ditmarsch, and N.~Alechina.
\newblock Groups versus coalitions: On the relative expressivity of {GAL} and
  {CAL}.
\newblock In {\em Proc.\ of the 18th {AAMAS}}, pages 953--961, 2019.

\bibitem{frenchetal:2008}
T.~French and H.~van Ditmarsch.
\newblock Undecidability for arbitrary public announcement logic.
\newblock In {\em Advances in Modal Logic 7}, pages 23--42, London, 2008.
  College Publications.

\bibitem{Galimullin19}
R.~Galimullin.
\newblock {\em Coalition announcements}.
\newblock PhD thesis, University of Nottingham, {UK}, 2019.

\bibitem{gerbrandyetal:1997}
J.~Gerbrandy and W.~Groeneveld.
\newblock Reasoning about information change.
\newblock {\em Journal of Logic, Language, and Information}, 6:147--169, 1997.

\bibitem{hales2013arbitrary}
J.~Hales.
\newblock Arbitrary action model logic and action model synthesis.
\newblock In {\em Proc.\ of 28th {LICS}}, pages 253--262. IEEE, 2013.

\bibitem{hales.aiml:2012}
J.~Hales, T.~French, and R.~Davies.
\newblock Refinement quantified logics of knowledge and belief for multiple
  agents.
\newblock In {\em Advances in Modal Logic 9}, pages 317--338. College
  Publications, 2012.

\bibitem{halpernzuck:1992}
J.~Halpern and L.~Zuck.
\newblock A little knowledge goes a long way: Knowledge-based derivations and
  correctness proofs for a family of protocols.
\newblock {\em J. ACM}, 39(3):449--478, 1992.

\bibitem{HennessyM85}
M.~Hennessy and R.~Milner.
\newblock Algebraic laws for nondeterminism and concurrency.
\newblock {\em J. {ACM}}, 32(1):137--161, 1985.

\bibitem{plaza:1989}
J.~Plaza.
\newblock Logics of public communications.
\newblock In {\em Proc.\ of the 4th ISMIS}, pages 201--216. Oak Ridge National
  Laboratory, 1989.

\bibitem{stulpetal:2002}
F.~Stulp and R.~Verbrugge.
\newblock A knowledge-based algorithm for the internet transmission control
  protocol ({TCP}).
\newblock {\em Bulletin of Economic Research}, 54(1):69--94, 2002.

\bibitem{blockchain}
M.~Swan.
\newblock {\em Blockchain: Blueprint for a New Economy}.
\newblock O'Reilly, 2015.

\bibitem{jfak.odds:1998}
J.~van Benthem.
\newblock Dynamic odds and ends.
\newblock Technical report, University of Amsterdam, 1998.
\newblock {ILLC} Research Report ML-1998-08.

\bibitem{jfak.lonely:2006}
J.~van Benthem.
\newblock One is a lonely number: on the logic of communication.
\newblock In {\em Logic colloquium 2002. Lecture Notes in Logic, Vol. 27},
  pages 96--129. A.K. Peters, 2006.

\bibitem{meyden:2004}
R.~van~der Meyden and K.~Su.
\newblock Symbolic model checking the knowledge of the dining cryptographers.
\newblock In {\em Proc.\ of the 17th IEEE workshop on Computer Security
  Foundations}, pages 280--291, 2004.

\bibitem{hvd.studlog:2003}
H.~van Ditmarsch.
\newblock The {R}ussian cards problem.
\newblock {\em Studia Logica}, 75:31--62, 2003.

\bibitem{hvdetal.jlc:2014}
H.~van Ditmarsch, D.~Fern{\'{a}}ndez{-}Duque, and W.~van~der Hoek.
\newblock On the definability of simulation and bisimulation in epistemic
  logic.
\newblock {\em J. Log. Comput.}, 24(6):1209--1227, 2014.

\bibitem{hvdetal.bapal:2017}
H.~van Ditmarsch and T.~French.
\newblock Quantifying over boolean announcements.
\newblock \url{https://arxiv.org/abs/1712.05310}, 2017.

\bibitem{papaldec:2018}
H.~van Ditmarsch, T.~French, and J.~Hales.
\newblock Decidability of arbitrary positive announcement logic.
\newblock Manuscript, 2018.

\bibitem{hvdetal.synthese:2006}
H.~van Ditmarsch and B.~Kooi.
\newblock The secret of my success.
\newblock {\em Synthese}, 151:201--232, 2006.

\bibitem{hvdetal.undecidable:2017}
H.~van Ditmarsch, W.~van~der Hoek, and L.~Kuijer.
\newblock The undecidability of arbitrary arrow update logic.
\newblock {\em Theor. Comput. Sci.}, 693:1--12, 2017.

\end{thebibliography}

\end{document}